\documentclass[11pt,a4paper,fleqn]{scrartcl}
\RequirePackage{amsthm,amsmath,natbib}
\RequirePackage{amssymb,amstext}
\RequirePackage{hypernat}
\usepackage[pdfpagemode=UseOutlines ,plainpages=false
,hypertexnames=false ,pdfpagelabels ,hyperindex=true,colorlinks=true]{hyperref}
	\makeatletter%
	\Hy@breaklinkstrue%
	\makeatother%
\usepackage{color}
\definecolor{darkred}{rgb}{0.6,0.0,0.1}
\definecolor{darkgreen}{rgb}{0,0.5,0}
\definecolor{darkblue}{rgb}{0,0,0.5}
\hypersetup{colorlinks ,linkcolor=darkblue ,filecolor=darkgreen
,urlcolor=darkblue ,citecolor=black}

\usepackage{url}
\usepackage{verbatim}
\usepackage{ bm}
\usepackage{graphicx}

\renewcommand{\cite}{\citet}
\usepackage{abrev-package}

\definecolor{dgreen}{rgb}{0,0.5,0}
\definecolor{dblue}{rgb}{0,0,0.9}
\definecolor{dred}{rgb}{0.6,0.0,0.1}
\definecolor{dgold}{rgb}{0.5,0.3,0.0}
\definecolor{dvio}{rgb}{0.6,0.3,0.5}
\definecolor{gray}{rgb}{0.5,0.5,0.5}

\theoremstyle{mysc}\newtheorem{prop}{Proposition}[section]
\theoremstyle{mysc}
\theoremstyle{mysc}
\theoremstyle{mysc}\newtheorem{theo}[prop]{Theorem}
\theoremstyle{mysc}\newtheorem{defin}{Definition}[section]
\theoremstyle{mysc}\newtheorem{lem}[prop]{Lemma}
\theoremstyle{myex}\newtheorem{rem}{Remark}[section]
\theoremstyle{myex}\newtheorem{example}{Example}[section]
\theoremstyle{myex}

\newtheorem{assA}{Assumption}

\theoremstyle{mysc}
\theoremstyle{mysc}
\theoremstyle{mysc}

\numberwithin{equation}{section}

\makeatletter%
\def\@fnsymbol#1{\ensuremath{\ifcase#1\or * \or \star \or 1 \or 2\or 3\or  , \or
g\or h\or i\else\@ctrerr\fi}}%
\makeatother%

 \author{\textsc{ Christoph Breunig}\thanks{Humboldt-Universit\"at zu Berlin, Spandauer Stra\ss e 1, 10178 Berlin, Germany, e-mail: christoph.breunig@hu-berlin.de}\\
\textit{ \small Humboldt-Universit\"at zu Berlin}}

\title{{Goodness-of-Fit Tests based on Series Estimators in Nonparametric Instrumental Regression}\thanks{
This paper derives from my doctoral dissertation, completed under the guidance of Enno Mammen.  I would like to thank two anonymous referees for comments and suggestions that greatly improved the paper. I also benefited from helpful comments by Jan Johannes, James Stock, Federico Crudu, and Petyo Bonev.
This work was supported by the DFG-SNF research group FOR916.}}

\date{March 2, 2015} 
\begin{document}
\maketitle
\vskip 1cm
\begin{abstract}
This paper proposes several tests of restricted specification in  nonparametric instrumental regression.
Based on series estimators, test statistics are established that allow for tests of the general model against a parametric or nonparametric specification as well as a test of exogeneity of the vector of regressors.
The tests' asymptotic distributions under correct specification are derived and their consistency against any alternative model is shown. Under a sequence of local alternative hypotheses, the asymptotic distributions of the tests is derived.  Moreover, uniform consistency is established over a class of alternatives whose distance to the null hypothesis shrinks appropriately as the sample size increases. A Monte Carlo study examines finite sample performance of the test statistics.

\end{abstract}
\begin{tabbing}
\> estimation, hypothesis testing, local alternative, uniform consistency.\\[.1ex]
\end{tabbing}
\section{Introduction}
While parametric instrumental variables estimators are widely used in econometrics, its nonparametric extension has not been introduced until the last decade. The study of nonparametric instrumental regression models was initiated by \cite{F03eswc} and \cite{NP03econometrica}.
In these models, given a scalar dependent variable $Y$, a vector of regressors $Z$, and a vector of instrumental variables $W$, the structural function $\sol$ satisfies
	\begin{equation}
	\label{model:niv}
	Y = \varphi(Z) + U \quad\mbox{ with }\quad\Ex[U|W]= 0
	\end{equation}
for an error term $U$. Here, $Z$ contains potentially  endogenous entries, that is, $\Ex[U|Z]$ may not be zero. Model \eqref{model:niv} does not involve the \textit{a priori} assumption that the structural function is known up to finitely many parameters. By considering this nonparametric model, we minimize the likelihood of misspecification. On the other hand, implementing the nonparametric instrumental regression model can be challenging.

Nonparametric instrumental regression models have attracted increasing attention in the econometric literature.
For example, 
\cite{AC03econometrica},  \cite{BCK07econometrica}, \cite{ChenReiss2008}, \cite{NP03econometrica} or \cite{Johannes2009} consider sieve minimum distance estimators of $\sol$, while  \cite{DFR02}, \cite{HallHorowitz2005}, \cite{GS06}
or \cite{FJVB07} study penalized least squares estimators.
When the methods of analysis are widened to include nonparametric techniques, one must confront two mayor challenges. First, identification in model \eqref{model:niv} requires far stronger assumptions about the instrumental variables than for the parametric case (cf. \cite{NP03econometrica}). Second, the accuracy of any estimator of $\sol$ can be low, even for large sample sizes.
More precisely, \cite{ChenReiss2008} showed that for a large class of joint distributions of $(Z,W)$ only logarithmic rates of convergence can be obtained. The reason for this slow convergence is that model \eqref{model:niv} leads to an inverse problem which is \textit{ill posed} in general, that is, the solution does not depend continuously on the data.

In light of the difficulties of estimating the nonparametric function $\sol$ in model \eqref{model:niv}, the need for statistically justified model simplifications is paramount.
We do not face an ill posed inverse problem if a parametric structure of $\sol$ or exogeneity of $Z$ can be justified. If these model simplifications are not supported by the data, one might still be interested in whether a smooth solution to model \eqref{model:niv} exists and if some regressors could be omitted from the structural function $\sol$. These model simplifications have important potential since they might increase the accuracy of estimators of $\sol$ or lower the required conditions imposed on  the instrumental variables to ensure identification.

In this work we present a new family of goodness-of-fit statistics which allows for several restricted specification tests of the model \eqref{model:niv}. Our method can be used for testing either a parametric or nonparametric specification. In addition, we perform a test of exogeneity and of dimension reduction of the vector of regressors $Z$, that is, whether certain regressors can be omitted from the structural function $\sol$.
By a withdrawal of  regressors which are independent of the instrument, identification in the restricted model might be possible although $\sol$ is not identified in the original model  \eqref{model:niv}.

There is a large literature concerning hypothesis testing of restricted specification of regression. In the context of conditional moment equation, \cite{Donald2003} and \cite{Tripathi2003} make use of empirical likelihood methods to  test parametric restrictions of the structural function. In addition,   \cite{Santos12} allows for different hypothesis tests, such as a test of homogeneity.
Based on kernel techniques, \cite{Horowitz2006}, \cite{Blundell2007}, and \cite{Horowitz2011} propose test statistics in which an additional smoothing step (on the exogenous entries of $Z$) is carried out. 
 \cite{Horowitz2006}  considers a parametric specification test.
\cite{Blundell2007} establish a consistent test of exogeneity of the vector of regressors $Z$, whereas \cite{Horowitz2011} tests whether the endogenous part of $Z$ can be omitted from $\sol$. \cite{GS07} and \cite{Horowitz2009} develop nonparametric specification tests in an instrumental regression model. We like to emphasize that their test cannot be applied to model \eqref{model:niv} where some entries of $Z$ might be exogenous.

Our testing procedure is entirely based on series estimation and hence is easy to implement. We use approximating functions to estimate the conditional moment restriction implied by the model \eqref{model:niv} where $\sol$ is replaced by an estimator under each conjectured hypothesis.
It is worth noting that by our methodology we can omit some assumptions typically found in related literature, such as smoothness conditions on the joint distribution of $(Z,W)$. In addition, a Monte Carlo indicates that the finite sample power of our tests exceed that of existing tests.

The paper is organized as follows. In Section 2, we start with a simple hypothesis test, that is, whether $\sol$ coincides with a known function $\sol_0$. We obtain the test's asymptotic distribution under the null hypothesis and its consistency  against any fixed alternative model. Moreover, we judge its power by considering linear local alternatives and establish uniform consistency over a class of functions. In Sections 3--5 we consider a parametric specification test, a test of exogeneity, and a nonparametric specification test. The goodness-of-fit statistics are obtained by replacing $\sol_0$ in the statistic of Section 2 by an appropriate estimator. 
In each case, the asymptotic distribution under correct specification and power statements against alternative models are derived. In Section 6, we investigate the finite sample properties of our tests by Monte Carlo simulations. All proofs can be found in the appendix.

\section{A simple hypothesis test}
In this section, we propose a goodness-of-fit statistic for testing the hypothesis $ H_0:\,\sol=\sol_0$, where $\sol_0$ is a known function, against the alternative $\sol\neq\sol_0$. We develop a test statistic based on $\cL^2$ distance. As we will see in the following chapters, it is sufficient to replace $\sol_0$ by an appropriate estimator to allow for tests of the general model against other specifications.
We first give basic assumptions, then obtain the asymptotic distribution of the proposed statistic, and further discuss its power and consistency properties.

\subsection{Assumptions and notation.}
\paragraph{The model revisited}
The nonparametric instrumental regression model \eqref{model:niv} leads to a linear operator equation. To be more precise, let us introduce the conditional expectation operator $T\phi:=\Ex[\phi(Z)|W]$ mapping $\cL^2_Z=\{\phi:\,\Ex|\phi(Z)|^2<\infty\}$ to $\cL^2_W=\{\psi:\, \Ex|\psi(W)|^2<\infty\}$. Consequently, model \eqref{model:niv} can be written as
	\begin{equation}
	\label{model:EE:t}
	g= T\sol
	\end{equation}
        where the function $g:=\Ex[Y | W]$ belongs to $\cL^2_W$.
Throughout the paper we assume that an iid. $n$-sample of $(Y,Z,W)$ from the model \eqref{model:niv} is available.
\paragraph{Assumptions.}
Our test statistic based on  a sequence of approximating functions $\{\basW_l\}_{l\geq 1}$ in $\cL^2_W$. Let $\cW$ denote the support of $W$ and the marginal density of $W$ by $p_W$.
Let $\nu$ be a probability density function that is strictly positive on $\cW$. We assume throughout the paper that $\{\basW_l\}_{l\geq 1}$ forms an orthonormal basis in $\cL_\nu^2(\mathbb R^{d_w}):=\{\phi:\,\int \phi^2(s)\nu(s)ds<\infty\}$ where $d_w$ denotes the dimension of $W$.  For instance, if $\cW\subset[a,b]$ then a natural choice of $\nu$ would be $\nu(w)=1/(b-a)$ for $w\in[a,b]$ and zero otherwise.
\begin{assA}\label{Ass:A1}
There exist constants $\eta_f, \eta_p\geqslant 1$ such that
(i) $ \sup_{l\geq 1}\int|f_l(s)|^4\nu(s)ds\leqslant \eta_f$ and (ii) $\sup_{w\in\cW}\big\{ p_W(w)/\nu(w)\big\}\leq \eta_p$ with  $\nu$ being strictly positive on $\cW$.
\end{assA}
 Assumption \ref{Ass:A1} $(i)$ restricts the magnitude of the approximating functions $\{f_j\}_{j\geq1}$ which is necessary for our proof to determine the asymptotic behavior of our test statistic. This assumption holds for sufficiently large $\eta_f$ if the basis $\{f_l\}_{l\geq1}$ is uniformly bounded, such as  trigonometric bases. Moreover,  Assumption \ref{Ass:A1} $(i)$ is satisfied by Hermite polynomials. Assumption \ref{Ass:A1} $(ii)$ is satisfied if, for instance, $p_W/\nu$ is continuous and $\cW$ is compact.

The results derived below involve assumptions on the conditional moments of the random variables $U$ given $W$ gathered in the following assumption.
\begin{assA}\label{Ass:A2}
 There exists a constant $\sigma>0$ such that $\Ex[U^4|W]\leq \sigma^4$.
\end{assA} 
The conditional moment condition on the error term $U$ helps to establish the asymptotic distribution of our test statistics.
 The following assumption ensures identification of $\sol$ in the model \eqref{model:EE:t}.
\begin{assA}\label{Ass:A3}
 The conditional expectation operator $T$ is nonsingular.
\end{assA} 
Under Assumption \ref{Ass:A3}, the hypothesis $H_0$ is equivalent to $g=T\sol_0$ which is used to construct our test statistic below. Note that the asymptotic results under null hypotheses considered in Sections 2--4 hold true even if $T$ is singular. If Assumption \ref{Ass:A3} fails, however, our test has no power against alternative models whose structural function satisfies $\sol=\sol_0+\delta$ with $\delta$ belonging to the null space of $T$.

We will see below that the power of our test can be increased by carrying out an additional smoothing step.
Therefore, we introduce a smoothing operator $L$ mapping  $\cL_W^2$ to $\cL_W^2$. 
In contrast to the unknown conditional expectation operator $\op$, which has to be estimated, the operator $L$ can be chosen by the econometrician.
Let $L$ have an eigenvalue decomposition given by $\{\tau_j^{1/2},f_j\}_{j\geq 1}$. We allow in this paper for a wide range of smoothing operators. In particular, $L$ may be the identity operator, that is, no smoothing step is carried out. We only require the following condition on the operator $L$ determined by the sequence of eigenvalues $\tau=(\tau_j)_{j\geq 1}$.
\begin{assA}\label{Ass:A4}
 The weighting sequence $\tau$ is positive, nonincreasing, and satisfies $\tau_1=1$.
\end{assA} 
Assumption \ref{Ass:A4} ensures that the operator $L$ is nonsingular.
\begin{rem}
 \cite{Horowitz2006}, \cite{Blundell2007}, and \cite{Horowitz2011} consider as a smoothing operator a Fredholm integral operator, that is,  $L\phi(s)=\int_0^1\ell(s,t)\phi(t)dt$ for some  function $\phi\in \cL^2[0,1]=\{\phi:\,\int_0^1\phi^2(s)ds<\infty\}$ and some kernel function $\ell:[0,1]^2\to\mathbb{R}$. In order to ensure $L\phi\in \cL^2[0,1]$ it is sufficient to assume $\int_0^1\int_0^1|\ell(s,t)|^2dsdt<\infty$. Let $\{\tau_j^{1/2},f_j\}_{j\geq 1}$ be the eigenvalue decomposition of $L$. By Parseval's identity
\begin{equation*}
 \int_0^1\int_0^1|\ell(s,t)|^2dsdt=\int_0^1\sum_{j=1}^\infty\tau_j|f_j(s)|^2ds=\sum_{j=1}^\infty\tau_j
\end{equation*}
where the right hand side is only finite if the sequence $\tau$ decays sufficiently fast.
In our case, if we apply a smoothing operator $L$ with $\sum_{j=1}^\infty\tau_j<\infty$ then our test statistics converges also to a weighted series of chi-squared random variables. In addition, we allow for a milder degree of smoothing or no smoothing at all and show below that then asymptotic normality of our test statistics can be obtained.
\hfill$\square$\end{rem}

\paragraph{Notation.}
 For a matrix $A$ we denote its transposed by $A^t$, its inverse by $A^{-1}$,  and its generalized inverse by $A^{-}$. The euclidean norm is denoted by $\|\cdot\|$ which in case of a matrix denotes the spectral norm, that is $\|A\|=(\text{trace}(A^tA))^{1/2}$. The norms on $L_Z^2$ and $L_W^2$ are denoted by $\|\phi\|_Z^2:=\Ex|\phi(Z)|^2$ for $\phi\in L_Z^2$ and $\|\psi\|_W^2:=\Ex|\psi(W)|^2$ for $\psi\in L_W^2$.
The $k\times k$ identity matrix  is denoted by $I_k$. For a vector $V$ we write  diag$(V)$ for the diagonal matrix with diagonal elements being the values of $V$.
 Moreover, $\basZ_\um(Z)$ and $\basW_\um(W)$ denote random vectors with entries $\basZ_j(Z)$ and $\basW_j(W)$, $1\leq j\leq m$, respectively. For any weighting sequence $w$ we introduce vectors $\basZ_\um^w(Z)$ and $\basW_\um^w(W)$ with entries $\basZ_j^w(Z)=\sqrt{w_j}e_j(Z)$ and $\basW_j^w(W)=\sqrt{w_j}f_j(W)$, $1\leq j\leq m$. 
We write $a_n\sim b_n$ when there exist
constants $c,c' > 0$ such that $c b_n\leq a_n\leq c' b_n$ for all sufficiently large $n$.

\subsection{The test statistic and its asymptotic distribution} 
Nonsingularity of the conditional expectation operator $T$  and
the smoothing operator $L$  implies that the null hypothesis $H_0$ is equivalent to $L(g-T\sol_0)=0$. Note that $\|L(g-T\sol_0)\|_W=0$ if and only if $\int\big|L(g-T\sol_0)(w)p_W(w)/\nu(w)\big|^2\nu(w)dw=0$ since the Lebesgue measure $\nu$ is strictly positive on $\cW$. Moreover, since $\{f_j\}_{j\geq 1}$ is an orthonormal basis with respect to $\nu$ we obtain by Parseval's identity
\begin{equation}
 \int\big|L(g-T\sol_0)(w)p_W(w)/\nu(w)\big|^2\nu(w)dw=\sum_{j=1}^\infty\Ex[(g-\op\sol_0)(W)f_j^\tau(W)]^2.
\end{equation}
 Now we truncate the infinite sum at some integer $\m$ which grows with the sample size $n$. This ensures consistency of our testing procedure.
Further, replacing the expectation by sample mean we obtain our test statistic
\begin{equation}\label{sn}
 S_n:=\sum_{j=1}^\m\tau_j\big|n^{-1}\sum_{i=1}^n(Y_i-\sol_0(Z_i))f_j(W_i)\big|^2.
\end{equation}
We reject the hypothesis $H_0$ if $n\,S_n$ becomes too large.
 When no additional smoothing is carried out, that is, $L$ is the identity operator, then $\tau_j=1$ for all $j\geq 1$.
 To achieve asymptotic normality we need to standardize our test statistic $S_n$ by appropriate mean and variance, which we introduce in the following definition.
\begin{defin}\label{def:sigma}
For all $m\geq 1$ let  $\varSigma_m$ be the covariance matrix of the random vector $Uf_\um^\tau(W)$ with entries ${\textsl s}_{jl}=\Ex\big[U^2f_j^\tau(W)f_{l}^\tau(W)\big]$, $1\leq j,l\leq m$.
Then the trace and the Frobenius norm of $\varSigma_m$ are respectively denoted by
\begin{equation*}
 \mu_m:=\sum_{j=1}^m {\textsl s}_{jj}\quad\text{and}\quad\varsigma_m:=\Big(\sum_{j,\,l=1}^m {\textsl s}_{jl}^2\Big)^{1/2}.
\end{equation*}
\end{defin}

Indeed the next result shows that $nS_n$ after standardization is asymptotically normally distributed if $\m$ increases appropriately as the sample size $n$ tends to infinity.

\begin{theo}\label{theo:norm:ind:main}
 Let  Assumptions \ref{Ass:A1}--\ref{Ass:A4} hold true. If   $\m$ satisfies
\begin{equation}\label{cond:theo:norm}
\varsigma_\m^{-1}=o(1)\quad \text{and}\quad \Big(\sum_{j=1}^\m\tau_j\Big)^3=o(n)
\end{equation}
 then under $H_0$
\begin{equation*}
 (\sqrt 2\varsigma_\m)^{-1}\big(n\,S_n-\mu_\m\big)\stackrel{d}{\rightarrow} \cN(0,1).
\end{equation*}
\end{theo}
\begin{rem}
Since $\varsigma_\m^2\leq \eta_p\,\sigma^4\big(\sum_{j=1}^\m\tau_j\big)^2$ (cf. proof of Theorem \ref{theo:norm:ind:main:tau}) condition $\varsigma_\m^{-1}=o(1)$ implies that $\sum_{j=1}^\m\tau_j$ tends to infinity as $n$ increases.
Moreover, from condition \eqref{cond:theo:norm} we see that by choosing a stronger decaying sequence $\tau$ the parameter $\m$ may be chosen larger. From the following theorem we see that if $\sum_{j=1}^\m\tau_j=O(1)$ only $m_n^{-1}=o(1)$ is required.
\hfill$\square$\end{rem}

In the following result, we establish the asymptotic distribution of our test when the sequence of weights $\tau$ may have a stronger decay than in Theorem \ref{theo:norm:ind:main}, that is, we consider the case where $\tau$ satisfies $\sum_{j=1}^\m\tau_j=O(1)$. This holds, for instance, if the sequence $\tau$ satisfies $\tau_j\sim j^{-(1+\varepsilon)}$ for any $\varepsilon>0$. In this case, the asymptotic distribution changes and additional definitions have to be made.
Let  $\varSigma$ be the covariance matrix of the infinite dimensional centered vector $\big(Uf_j^\tau(W)\big)_{j\geq 1}$. The ordered eigenvalues of $\varSigma$ are denoted by $(\lambda_j)_{j\geq1}$. Below, we introduce a sequence $\{\chi_{1j}^2\}_{j\geq 1}$ of independent random variables that are distributed as chi-square with one degree of freedom.

\begin{theo}\label{theo:norm:ind:main:tau}
Let  Assumptions \ref{Ass:A1}--\ref{Ass:A4} hold true. If   $\m$ satisfies
\begin{equation}\label{cond:theo:norm:tau}
\sum_{j=1}^\m\tau_j=O(1)\quad\text{and}\quad m_n^{-1}=o(1)
\end{equation}
 then  under $H_0$
\begin{equation*}
n\,S_n\stackrel{d}{\rightarrow} \sum_{j=1}^\infty \lambda_j\,\chi_{1j}^2.
\end{equation*}
\end{theo}

\begin{rem}[Estimation of Critical Values]
The asymptotic results of Theorem \ref{theo:norm:ind:main} and \ref{theo:norm:ind:main:tau} depend on unknown population quantities. As we see in the following, the critical values can be easily estimated. 
Let $\bold W_m (\tau)$ denote a $n\times m$ matrix with entries $f_j^\tau(W_i)$ for $1\leq i\leq n$ and $1\leq j\leq m$. Moreover, 
$\bold U_n=(Y_1-\sol_0(Z_1),\dots,Y_n-\sol_0(Z_n))^t$. In the setting of Theorem \ref{theo:norm:ind:main}, we replace $\varSigma_m$ by 
\begin{equation*}
 \widehat\varSigma_m:=n^{-1}\bold W_m(\tau)^t\,\text{diag}(\bold U_n)^2\,\bold W_m(\tau).
\end{equation*}
Now the asymptotic result of Theorem \ref{theo:norm:ind:main} continues to hold if we replace $\varsigma_\m$ by the  Frobenius norm of $\widehat\varSigma_\m$ and $\mu_\m$ by the trace of $\widehat\varSigma_\m$. In the setting of Theorem \ref{theo:norm:ind:main:tau}, the asymptotic distribution is not pivotal and has to approximated. First, the difference of critical values between $\sum_{j=1}^\infty \lambda_j\chi_{1j}^2$ and the truncated sum $\sum_{j=1}^{M_n}\lambda_j\,\chi_{1j}^2$ converges to zero if the integer $M_n>0$ tends to infinity (depending on $n$). Second, replace $(\lambda_j)_{1\leq j\leq M_n}$ by $(\widehat\lambda_j)_{1\leq j\leq M_n}$ which are the ordered eigenvalues of $\widehat\varSigma_{M_n}$. Observe $\max_{1\leq j\leq M_n}|\widehat\lambda_j-\lambda_j|=\|\widehat\varSigma_{M_n}-\varSigma_{M_n}\|=O(M_n n^{-1/2})$ almost surely. Hence, the critical values of $\sum_{j=1}^{M_n}\widehat\lambda_j\,\chi_{1j}^2$ converge 
in probability to the ones of the limiting 
distribution of $n\,
S_n$ if $M_n=o(\sqrt{n})$.
\hfill$\square$\end{rem}

\subsection{Limiting behavior under local alternatives.}
Let us study the power of the test statistic $S_n$, that is, the probability to reject a false hypothesis, against a sequence of linear local alternatives that tends to zero as $n\to\infty$. It is shown that the power of our tests essentially relies on the choice of the weighting sequence $\tau$.

Let us start with the case $\varsigma_\m^{-1}=o(1)$. We consider  the following sequence of linear local alternatives 
\begin{equation}\label{loc:alt:ind}
 Y=\sol_0(Z)+\varsigma_\m^{1/2} n^{-1/2}\delta(Z)+U
\end{equation}
for some function $\delta\in \cL_Z^4:=\{\phi:\Ex|\phi(Z)|^4<\infty\}$.
The next result establishes asymptotic normality for the standardized test statistic $S_n$. Let us denote $\delta_j:=\sqrt{\tau_j}\,\Ex[\delta(Z)f_j(W)]$.
\begin{prop}\label{coro:norm:ind:main}Given the conditions of Theorem \ref{theo:norm:ind:main} it holds under \eqref{loc:alt:ind}
\begin{equation*}
 (\sqrt 2\varsigma_\m)^{-1}\big(n\,S_n-\mu_\m\big)\stackrel{d}{\rightarrow}\cN\Big(2^{-1/2}\sum_{j= 1}^\infty\delta_j^2,1\Big).
\end{equation*}
\end{prop}

As we see below the test statistic $S_n$ has power advantages if $\sum_{j=1}^\m\tau_j=O(1)$.
Let us consider the sequence of linear local alternatives 
\begin{equation}\label{loc:alt:ind:2}
 Y=\sol_0(Z)+n^{-1/2}\delta(Z)+U
\end{equation}
for some function $\delta\in \cL_Z^4$. For the next result, the sequence $\{\chi_{1j}^2(\delta_j/\lambda_j)\}_{j\geq 1}$  denotes independent random variables that are distributed as non-central chi-square with one degree of freedom and non-centrality parameters $\delta_j/\lambda_j$.

\begin{prop}\label{coro:norm:ind:main:tau}Given the conditions of Theorem \ref{theo:norm:ind:main:tau} it holds under \eqref{loc:alt:ind:2}
\begin{equation*}
 n\, S_n\stackrel{d}{\rightarrow}\sum_{j=1}^\infty\lambda_j\,\chi_{1j}^2(\delta_j/\lambda_j).
\end{equation*}
\end{prop}

\begin{rem}
We see from Proposition \ref{coro:norm:ind:main} that our test  can detect linear alternatives at a rate $\varsigma_\m^{1/2} n^{-1/2}$.  On the other hand, if $\sum_{j=1}^\m\tau_j=O(1)$ then $S_n$ can detect local linear alternatives  at the faster rate $n^{-1/2}$. But still our test with $L=\text{Id}$ can have better power against certain smooth classes of alternatives as illustrated  by \cite{Hong95} and \cite{Horowitz2001}. Indeed, the next subsection shows that additional smoothing changes the class of alternatives over which uniform consistency can be obtained. 
\hfill$\square$
\end{rem}

\subsection{Consistency}
In this subsection, we establish consistency against a fixed alternative and uniform consistency of our test over appropriate function classes.
Let us first consider the case of a fixed alternative. We assume that $H_0$ does not hold, that is, $\PP(\sol=\sol_0)<1$. The following proposition shows that our test has the ability to reject a false null hypothesis with probability $1$ as the sample size grows to infinity.

The consistency properties require the following additional assumption.

\begin{assA}\label{Ass:power}
 (i) The function $p_W/\nu$ is uniformly bounded away from zero. (ii) There exists a constant $\sigma_o>0$ such that $\Ex[U^2|W]\geq \sigma_o^2$. 
\end{assA}
Assumption \ref{Ass:power} $(i)$ implies that $\|LT(\sol-\sol_0)\|_W>0$ for any structural function $\sol$ in the alternative. Further, Assumption \ref{Ass:power} implies that $\sum_{j=1}^\m\tau_j^2=O(\varsigma_\m^2)$. 
\begin{prop}\label{prop:cons}Assume that $H_0$ does not hold. Let $\Ex|Y-\sol_0(Z)|^4<\infty$ and let Assumption  \ref{Ass:power} (i) hold true. Consider the sequence $(\alpha_n)_{n\geq 1}$ satisfying $\alpha_n=o(n\varsigma_\m^{-1})$. Under the conditions of Theorem \ref{theo:norm:ind:main} we have
\begin{gather*}
 \PP\Big((\sqrt 2\,\varsigma_\m)^{-1}\big( nS_n-\mu_\m\big)>\alpha_n\Big)=1+o(1).\label{prop:cons:1}
\end{gather*}
 Under the conditions of Theorem \ref{theo:norm:ind:main:tau} we have $\alpha_n=o(n)$ and
\begin{gather*}
 \PP\big(nS_n>\alpha_n\big)=1+o(1).\label{prop:cons:2}
\end{gather*}
\end{prop}

In the following, we specify a class of functions over which our test $S_n$ is uniformly consistent. This essentially implies that there are no alternative functions in this class over which our test has low power.
We show that our test is consistent uniformly over the class
\begin{equation*}\label{cgn}
 \cG_{n}^\sr=\Big\{\sol\in \cL_Z^2:\,\|LT(\sol-\sol_0)\|_W^2\geq \sr \, n^{-1}\varsigma_\m\text{ and}\sup_{z\in\cZ}|(\sol-\sol_0)(z)|^2\leq C\Big\}
\end{equation*}
where $C>0$ is a finite constant.
Clearly, if $H_0$ is false then $\|LT(\sol-\sol_0)\|_W^2\geq \sr\,\varsigma_\m n^{-1}$ for all sufficiently large $n$ and some $\rho>0$.
By Assumption \ref{Ass:A4} the sequence $\tau$ is nonincreasing sequence with $\tau_1=1$ and hence,  $\|LT(\sol-\sol_0)\|_W^2\leq \|T(\sol-\sol_0)\|_W^2\leq \|\sol-\sol_0\|_Z^2$ by Jensen's inequality. We conclude that $\cG_n^\sr$ contains all alternative functions whose $\cL_Z^2$-distance to the structural function $\sol_0$ is at least $n^{-1}\varsigma_\m$ within a constant. 
If the coefficients $\Ex[(\sol-\sol_0)(Z)f_j(W)]$ fluctuate for large $j$  then $\sol$ does not belong to $\cG_n^\sr$ if the decay of $\tau$ is too strong.
 On the other hand, if $\Ex[(\sol-\sol_0)(Z)f_j(W)]$ is sufficiently small for $j$ up to a finite constant then $\sol$ does not necessarily belong to $\cG_{n}^\sr$ with $\tau$ having a slow decay.
For the next result let $q_{1\alpha}$ and $q_{2\alpha}$ denote the $1-\alpha$ quantile of $\cN(0,1)$ and $\sum_{j=1}^\infty\lambda_j\,\chi_{1j}^2$, respectively.
\begin{prop}\label{prop:unf:cons}
Let Assumption  \ref{Ass:power} be satisfied.
For any $\varepsilon>0$, any  $0<\alpha<1$, and any sufficiently large constant $\rho>0$ we have under the conditions of Theorem \ref{theo:norm:ind:main}  that
\begin{equation*}
  \lim_{n\to\infty}\inf_{\sol\in\cG_n^\sr}\PP\Big((\sqrt 2\,\varsigma_\m)^{-1}\big( nS_n-\mu_\m\big)>q_{1\alpha}\Big)\geq 1-\varepsilon,
\end{equation*}
while under the conditions of Theorem \ref{theo:norm:ind:main:tau}
\begin{equation*}
  \lim_{n\to\infty}\inf_{\sol\in\cG_n^\sr}\PP\Big(n\,S_n>q_{2\alpha}\Big)\geq 1-\varepsilon.
\end{equation*}
\end{prop}

\section{A parametric specification test}
In this section, we present a test whether the structural function $\sol$ is known up to a finite dimensional parameter.  Let $\varTheta$ be a compact subspace of $\mathbb R^k$ then we consider  the null hypothesis $H_{\text{p}}:$ there exists some $\vartheta\in\varTheta$ such that $\sol(\cdot)=\phi(\cdot,\vartheta)$ for a known function $\phi$. The alternative hypothesis is that there exists no $\vartheta\in\varTheta$ such that $\sol(\cdot)=\phi(\cdot,\vartheta)$ holds true.
\subsection{The test statistic and its asymptotic distribution}
Under Assumptions \ref{Ass:A3} and \ref{Ass:A4}, the null hypothesis $H_{\text{p}}$ is equivalent to $L(g-T\phi(\cdot,\vartheta))=0$ for some $\vartheta\in\varTheta$. 
Thereby, to verify $H_{\text{p}}$ we make use of the test statistic $S_n$ given in \eqref{sn} where $\sol_0$ is replaced by $\phi(\cdot,\hthet)$ with $\hthet$ being an estimator of $\vartheta$.
Hence, our test statistic for a parametric specification is given by 
\begin{equation*}
  S_n^{\text{p}}:=\sum_{j=1}^\m\tau_j\big|n^{-1}\sum_{i=1}^n\big(Y_i-\phi(Z_i,\hthet)\big)f_j(W_i)\big|^2.
\end{equation*}
If the test statistic $S_n^{\text{p}}$ becomes too large then $H_{\text{p}}$ has to be rejected.
To obtain asymptotic results for the statistic $S_n^{\text{p}}$ we require smoothness conditions of the function $\phi$ with respect to its second argument. 
Below we denote the vector of partial derivatives of $\phi$ with respect to $\vartheta=(\vartheta_1,\dots,\vartheta_k)^t$ by $\phi_\vartheta=(\phi_{\vartheta_l})_{1\leq l\leq k}$ and the matrix of second-order partial derivatives by $\phi_{\vartheta\vartheta}=(\phi_{\vartheta_j\vartheta_l})_{1\leq j,l\leq k}$.
\begin{assA}\label{Ass:par} (i) Let $\hthet$ be an estimator satisfying $\|\hthet-\thet\|=O_p(n^{-1/2})$ for some  $\thet\in{\text{int}}(\varTheta)$ with $\sol(\cdot)=\phi(\cdot,\vartheta_0)$ if $H_{\text{p}}$ holds true.
(ii) The function $\phi$ is twice partial differentiable with respect to its second argument.
There exists some constant $\eta_\phi\geqslant 1$   such that
\begin{equation*}
 \sup_{1\leq l\leq k}\Ex |\phi_{\vartheta_l}(Z,\thet)|^4\leqslant \eta_\phi\quad\text{and }\quad 
\sup_{1\leq j,l\leq k}\Ex \sup_{\theta\in\Theta}|\phi_{\vartheta_j\vartheta_l}(Z,\theta)|^4\leqslant \eta_\phi.
\end{equation*}
\end{assA}

The following proposition establishes asymptotic normality of $S_n^{\text{p}}$ after standardization.
\begin{theo}\label{prop:par}
Let  Assumptions \ref{Ass:A1}--\ref{Ass:A4} and \ref{Ass:par} hold true. If   $\m$ satisfies \eqref{cond:theo:norm}, then under $H_{\textsl{p}}$
\begin{equation*}
 (\sqrt 2\varsigma_\m)^{-1}\big(n\, S_n^{\textsl{p}}-\mu_\m\big)\stackrel{d}{\rightarrow} \cN(0,1).
\end{equation*}
\end{theo}
In the following theorem, we state the asymptotic distribution of $nS_n^{\textsl{p}}$ when $\sum_{j=1}^\m\tau_j=O(1)$. 
In this case, we assume that $\hthet$ satisfies under $H_{\textsl{p}}$
\begin{equation}\label{prop:par:cond}
 \sqrt n(\hthet-\thet)=n^{-1/2}\sum_{i=1}^n h_\uk(V_{i})+o_p(1)
\end{equation}
where $V_{i}:=(Y_i,Z_i,W_i,\thet)$ and $h_\uk(V_{i})=(h_1(V_i),\dots, h_k(V_i))^t$ where $h_j$, $1\leq j\leq k$,  are real valued functions. 
It is well known that this representation holds if $\hthet$ is the generalized method of moments estimator.
Let  $\varSigma^\textsl{p}$ be the covariance matrix of the infinite dimensional centered vector $\big(Uf_j^\tau(W)-\Ex[f_j^\tau(W)\phi_\vartheta(Z,\thet)^t]h_\uk(V)\big)_{j\geq 1}$. The ordered eigenvalues of $\varSigma^\textsl{p}$ are denoted by $(\lambda_j^\textsl{p})_{j\geq1}$.
\begin{theo}\label{prop:par:tau}
Let  Assumptions\ref{Ass:A1}--\ref{Ass:A4} and \ref{Ass:par} hold true. Assume that $H_{\textsl{p}}$ holds true and $\hthet$ satisfies condition \eqref{prop:par:cond} with $\Ex h_{j}(V)=0$ and $\Ex |h_{j}(V)|^4<\infty$, $1\leq j\leq k$. 
If   $\m$ satisfies \eqref{cond:theo:norm:tau}, then
\begin{equation*}
 n\,S_n^{\textsl{p}}\,\stackrel{d}{\rightarrow} \sum_{j=1}^\infty\lambda_j^\textsl{p}\,\chi_{1j}^2.
\end{equation*}
\end{theo}
\begin{rem}\label{est:crit:par}[Estimation of Critical Values]
 For the estimation of critical values of Theorem  \ref{prop:par} and \ref{prop:par:tau}, let us define $\bold U_n^{\textsl{p}}=\big(Y_1-\phi(Z_1,\hthet),\dots,Y_n-\phi(Z_n,\hthet)\big)^t$. We estimate the covariance matrix $\varSigma_m$ by 
\begin{equation*}
 \widehat\varSigma_m:=n^{-1}\,\bold W_m(\tau)^t\,\text{diag}(\bold U_n^{\textsl{p}})^2\,\bold W_m(\tau).
\end{equation*}
Now the asymptotic result of Theorem \ref{prop:par} continues to hold if we replace $\varsigma_\m$ by the  Frobenius norm of $\widehat\varSigma_\m$ and $\mu_\m$ by the trace of $\widehat\varSigma_\m$. In the setting of Theorem \ref{prop:par:tau}, we replace $\varSigma^\textsl{p}$ by a finite dimensional matrix.
Let $\bold A_k$ be a $n\times k$ matrix with entries $\phi_{\vartheta_l}(Z_i,\hthet)$ for $1\leq i\leq n$, $1\leq l\leq k$ and $\bold h_k(V)=\big(h_\uk(V_{1}),\dots,h_\uk(V_{n})\big)^t$. Then define $\bold V_k:=n^{-1}\bold h_k(V)\bold A_k^t $. Given a sufficiently large integer $M>0$ we estimate $\varSigma^\textsl{p}$ by 
\begin{equation*}
  \widehat\varSigma_M^{\textsl{p}}:=n^{-1}\bold W_M(\tau)^t
\Big(\text{diag}(\bold U_n^{\textsl{p}})-\bold V_k\Big)^t
\Big(\text{diag}(\bold U_n^{\textsl{p}})-\bold V_k\Big)\bold W_M(\tau).
\end{equation*}
Hence, we approximate $\sum_{j=1}^\infty \lambda_j\chi_{1j}^2$ by the finite sum $\sum_{j=1}^{M_n}\widehat\lambda_j^{\textsl{p}}\,\chi_{1j}^2$ where $(\widehat\lambda_j^{\textsl{p}})_{1\leq j\leq M_n}$ are the ordered eigenvalues of $\widehat\varSigma_{M_n}^{\textsl{p}}$. We have $\max_{1\leq j\leq M_n}|\widehat\lambda_j^{\textsl{p}}-\lambda_j^{\textsl{p}}|=o_p(1)$ if $M_n=o(\sqrt n)$. 
\hfill$\square$\end{rem}

\subsection{Limiting behavior under local alternatives and consistency.}
In the following, we study the power and consistency properties of the test statistic $S_n^{\text{p}}$.
In the following, we consider a sequence of linear local alternatives \eqref{loc:alt:ind} or \eqref{loc:alt:ind:2} with $\sol_0=\phi(\thet,\cdot)$.
Further, let $\delta_\perp$ denote the projection of $\delta$ onto the orthogonal complement of $\phi(\cdot,\thet)$; that is, $\Ex[\phi_\vartheta(Z,\thet)\delta_\perp(Z)]=0$. Let us denote $\delta_{j\perp}:=\sqrt{\tau_j}\,\Ex[\delta_\perp(Z)f_j(W)]$.

\begin{prop}\label{coro:norm:par}Let the conditions of Theorem \ref{prop:par} be satisfied. Then under \eqref{loc:alt:ind} with $\sol_0=\phi(\cdot,\thet)$
it holds
\begin{equation*}
 (\sqrt 2\varsigma_\m)^{-1}\big(n\,S_n^{\textsl{p}}-\mu_\m\big)\stackrel{d}{\rightarrow}\cN\Big(2^{-1/2}\sum_{j= 1}^\infty\delta_{j\perp}^2,1\Big).
\end{equation*}
Let  the conditions of Theorem \ref{prop:par:tau} be satisfied. Then under \eqref{loc:alt:ind:2} with $\sol_0=\phi(\cdot,\thet)$
it holds
\begin{equation*}
n\,S_n^{\textsl{p}}\stackrel{d}{\rightarrow}\sum_{j=1}^\infty\lambda_j^{\textsl{p}}\,\chi_{1j}^2(\delta_{j\perp}/\lambda_j^{\text{p}}).
\end{equation*}
\end{prop}

\begin{rem}
 Under  homoscedasticity, that is, $\Ex[U^2|W]=\sigma_o^2$, $W\sim\cU[0,1]$,  and $L=\text{Id}$ we see from Proposition \ref{coro:norm:par} that our test has the same power properties as the test of \cite{Hong95}. On the other hand, if $\sum_{j=1}^\m\tau_j=O(1)$ then our test can detect local linear alternatives  at a rate $n^{-1/2}$ as in \cite{Horowitz2006}, which decreases more quickly than the rate obtained by \cite{Tripathi2003}. \hfill$\square$
\end{rem} 
The next proposition establishes consistency of our test against a fixed alternative model. It is assumed that $H_\text p$ is false, that is, there exists no $\vartheta\in\varTheta$ such that $\sol(\cdot)=\phi(\cdot,\vartheta)$. In this situation, $\thet$ denotes the probability limit of the estimator $\hthet$.
\begin{prop}\label{prop:par:cons}Assume that $H_{\textsl{p}}$ does not hold. Let $\Ex|Y-\phi(Z,\thet)|^4<\infty$ and Assumption  \ref{Ass:power} (i) hold true. Let $(\alpha_n)_{n\geq 1}$ as in Proposition \ref{prop:cons}. Under the conditions of Theorem \ref{prop:par} we have
\begin{equation*}
 \PP\Big((\sqrt 2\,\varsigma_\m)^{-1}\big( nS_n^{\textsl{p}}-\mu_\m\big)>\alpha_n\Big)=1+o(1).
\end{equation*}
Given the conditions of Theorem \ref{prop:par:tau} it holds
\begin{equation*}
  \PP\big(n\,S_n^{\textsl{p}}>\alpha_n\big)=1+o(1).
\end{equation*}
\end{prop}

In the following, we show that $S_n^{\textsl{p}}$ is consistent uniformly over the function class
 \begin{multline*}\label{cgn}
 \cH_{n}^\sr=\Big\{\sol\in \cL_Z^2:\,\|LT(\sol-\phi(\cdot,\vartheta_0))\|_W^2\geq \sr \, n^{-1}\varsigma_\m
\text{ and}\sup_{z\in\cZ}|\sol(z)-\phi(z,\vartheta_0)|\leq C\Big\}
\end{multline*}
for some constant $C>0$ and $\thet$ denotes the probability limit of $\hthet$.
Similarly as in the previous section, it can be seen that
 $\cH_{n}^\sr$ only contains functions whose $\cL_Z^2$ distance to $\phi(\cdot,\vartheta_0)$ is at least $n^{-1}\varsigma_\m$ within a constant. For the next result let $q_{1\alpha}$ and $q_{2\alpha}$ denote the $1-\alpha$ quantile of $\cN(0,1)$ and $\sum_{j=1}^\infty\lambda_j^{\textsl{p}}\,\chi_{1j}^2$, respectively.
\begin{prop}\label{prop:par:unf:cons} Let Assumption  \ref{Ass:power} be satisfied.
For any $\varepsilon>0$, any  $0<\alpha<1$, and any sufficiently large constant $\rho>0$ we have under the conditions of Theorem \ref{prop:par} that
\begin{equation*}
  \lim_{n\to\infty}\inf_{\sol\in\cH_n^\sr}\PP\Big((\sqrt 2\,\varsigma_\m)^{-1}\big(n\,S_n^{\textsl{p}}-\mu_\m\big)>q_{1\alpha}\Big)\geq 1-\varepsilon,\label{prop:par:unf:cons:1}
\end{equation*}
whereas under the conditions of Theorem \ref{prop:par:tau} it holds
\begin{equation*}
 \lim_{n\to\infty}\inf_{\sol\in\cH_{n}^\sr}\PP\big(n\,S_n^{\textsl{p}}>q_{2\alpha}\big)\geq 1-\varepsilon.\label{prop:par:unf:cons:2}
\end{equation*}
\end{prop}

\section{A nonparametric test of exogeneity}
Endogeneity of regressors is a common problem in econometric applications. Falsely assuming exogeneity of the regressors leads to inconsistent estimators. On the other hand, treating exogenous regressors as if they were endogenous can lower the accuracy of estimation dramatically. In this section, we propose a test whether the vector of regressors $Z$ is exogenous, that is, $\Ex[U|Z]=0$ or equivalently $\sol(Z)=\Ex[Y|Z]$. In this section, let $\sol_0(Z)=\Ex[Y|Z]$ then the hypothesis under consideration is given by  $H_{\textsl{e}}:\, \sol=\sol_0$. The alternative hypothesis is that $\sol\neq \sol_0$. 
\subsection{The test statistic and its asymptotic distribution}
To establish a test of exogeneity, let us first introduce an estimator of the conditional mean of $Y$ given $Z$. This estimator is based on  a sequence of approximating functions $\{e_j\}_{j\geq 1}$ belonging to $\cL_Z^2$.
Further, let $\bold Z_k$ denote a $n\times k$ matrix with entries $e_j(Z_i)$ for $1\leq i\leq n$ and $1\leq j\leq k$. Moreover, let $\bold Y_n=(Y_1,\dots,Y_n)^t$.
Then we define the estimator 
\begin{equation}\label{est:ex}
\osol_k(\cdot):=
e_\uk(\cdot)^t\widehat \beta_k\quad\text{where }\quad \widehat \beta_k=(\bold Z_k^t\bold Z_k)^-\,\bold Z_k^t \bold Y_n.
\end{equation}
 In contrast to the parametric case we need to allow for $k$ tending to infinity as $n\to\infty$ in order to ensure consistency of the estimator $\osol_k$. Under conditions given below $\bold Z_\k^t\bold Z_\k$ will be nonsingular with probability approaching one and hence its generalized inverse will be the standard inverse. Note that the asymptotic behavior of the estimator $\osol_k$ was studied, for example, by \cite{Newey1997}. 

Under Assumptions \ref{Ass:A3} and \ref{Ass:A4}, the null hypothesis $H_{\textsl{e}}$ is equivalent to $L(g-T\sol_0)=0$. Consequently, our test of exogeneity of $Z$ is based on the goodness-of-fit statistic $S_n$ introduced in \eqref{sn} but where $\sol_0$ is replaced by the series estimator $\osol_\k$.
The proposed test statistic for $H_{\textsl{e}}$ is now given by 
\begin{equation*}
 S_n^{\textsl{e}}=\sum_{j=1}^\m\tau_j\big|n^{-1}\sum_{i=1}^n\big(Y_i-\osol_\k(Z_i)\big)f_j(W_i)\big|^2
\end{equation*}
where $\k$ and $\m$ tend to infinity as $n\to\infty$. The hypothesis of exogeneity of $Z$ has to be rejected if $S_n^{\textsl{e}}$ becomes too large. 

For controlling the bias of the estimator $\osol_\k$ we specify in the following a rate of approximation (cf. \cite{Newey1997}).
Let $\sw=(\sw_j)_{j\geq 1}$ be a nondecreasing sequence with $\sw_1=1$. We assume that $\sol_0$ belongs to
\begin{equation*}
 \cF_{\sw} := \big\{\phi\in \cL_Z^2: \sup_{z\in\cZ}|\phi(z)-e_\ukn(z)^t\beta_\k|^2=O(\sw_\k^{-1})\text{ for some } \beta_\k\in\mathbb R^\k\big\}.
\end{equation*}
Here, the sequence of weights $\sw$ measures the approximation error of $\sol_0$ with respect to the functions $\{e_j\}_{j\geq1}$.

\begin{assA}\label{Ass:ex} (i) Let $\sol_0\in\cF_\sw$ with nondecreasing sequence $\sw$ satisfying $j^2=o(\sw_j)$.
 (ii) There exists some constant $\eta_e\geqslant 1$ such that  $\sup_{z\in\cZ}\|e_\ukn(z)\|^2\leqslant \eta_e\k$.
 (iii) The smallest eigenvalue of $\Ex[e_\uk(Z)e_\uk(Z)^t]$ is bounded away from zero uniformly in $k$. (iv) $\Ex[U^2|Z]$ is bounded.
\end{assA}
Assumption \ref{Ass:ex} $(i)$ determines the required asymptotic behavior of the rate $\sw$. For splines and power series this assumption is satisfied if the number of continuous derivatives of $\sol_0$ divided by the dimension of $Z$ equals two. Assumption \ref{Ass:ex} $(ii)$ and $(iii)$ restrict the magnitude of the approximating functions $\{e_j\}_{j\geq1}$ and impose nonsingularity of their second moment matrix.

We are now in the position to proof the following asymptotic result for the standardized test statistic $S_n^\text e$. Here, a key requirement is that $\k=o(\varsigma_\m)$ implying that $\k=o(\sum_{j=1}^\m\tau_j)$ and, in particular, $\k=o(\m)$ if the smoothing operator $L$ is the identity.
\begin{theo}\label{prop:ex} Let  Assumptions \ref{Ass:A1}--\ref{Ass:A4} and \ref{Ass:ex} be satisfied. If
\begin{equation}\label{cond:prop:ex}
n=o(\sw_\k\varsigma_\m), \,\k=o(\varsigma_\m),\quad \text{and}\quad \Big(\sum_{j=1}^\m\tau_j\Big)^3=o(n)
\end{equation}then under $H_{\textsl{e}}$ it holds
\begin{equation*}
 (\sqrt 2\varsigma_\m)^{-1}\big(n\,S_n^{\textsl{e}}-\mu_\m\big)\stackrel{d}{\rightarrow} \cN(0,1).
\end{equation*}
\end{theo}

\begin{example}\label{exmp:cases}
Let $Z$ be continuously distributed with $\dim(Z)=r$ and set $L=\text{Id}$.
 Consider the polynomial case where $\sw_j\sim j^{2p/r}$ with $p>1$ and let $\m\sim n^\nu$ with $0<\nu<1/3$. Let Assumption \ref{Ass:power} hold true then $\sqrt \m=O(\varsigma_\m)$.
Hence, condition \eqref{cond:prop:ex} is satisfied if $\k\sim n^\kappa$ with
\begin{equation}\label{cond:prop:ex:cases}
 r(1-\nu/2)/(2p)<\kappa<\nu/2.
\end{equation}
This ensures that the bias of this estimator in the statistic $S_n^{\textsl{e}}$ is asymptotically negligible.
Note that condition \eqref{cond:prop:ex:cases} requires $2p>r\,(2/\nu-1)$. Hence, with a larger dimension $r$ also the smoothness of $\sol_0$ has to increase, reflecting the curse of dimensionality.
$\hfill\square$
\end{example}

The next result states an asymptotic distribution result for the statistic $S_n^{\textsl{e}}$ if $\sum_{j=1}^\m\tau_j=O(1)$.
Let  $\varSigma^\textsl{e}$ be the covariance matrix of the infinite dimensional centered vector $\big(U(f_j^\tau(W)-\sum_{l\geq 1}\Ex[f_j^\tau(W)e_l(Z)]e_l(Z))\big)_{j\geq 1}$.
The ordered eigenvalues of $\varSigma^\textsl{e}$ are denoted by $(\lambda_j^\textsl{e})_{j\geq1}$.
\begin{theo}\label{prop:ex:tau}Let  Assumptions \ref{Ass:A1}--\ref{Ass:A4} and \ref{Ass:ex} be satisfied.
If
\begin{equation}\label{prop:ex:tau:cond}
 \sum_{j=1}^\m\tau_j=O(1),\quad n=O(\sw_\k),\quad k_n^3=o(n),\text{ and}\quad m_n^{-1}=o(1)
\end{equation}
then under $H_{\textsl{e}}$ it holds
\begin{equation*}
 n\,S_n^{\textsl{e}}\stackrel{d}{\rightarrow} 
\sum_{j=1}^\infty\lambda_j^\textsl{e}\,\chi_{1j}^2.
\end{equation*}
\end{theo}
\begin{example}\label{exmp:cases:2}
Consider the setting of Example \ref{exmp:cases} but where the eigenvalues of $L$ satisfy $\tau_j\sim j^{-2}$. 
 Condition \eqref{prop:ex:tau:cond} is satisfied if $\m\sim n^\nu$ for some $\nu>0$ and $\k\sim n^\kappa$ with $r/(2p)<\kappa<1/3$.
Here, the required smoothness of $\sol_0$ is $p>3r/2$.
 In contrast to the setting of Theorem \ref{prop:ex}, the estimator of $\sol_0$ needs to be undersmoothed. 
 This ensures that the bias of this estimator in the statistic $S_n^{\textsl{e}}$ is asymptotically negligible.
$\hfill\square$
\end{example}
\begin{rem}
 In contrast to \cite{Blundell2007} no smoothness assumptions on the joint distribution of $(Z,W)$ is required here. In addition, we do not need any assumption that links the smoothness of the regression function $\sol_0$ to the smoothness of the joint density of $(Z,W)$.
\hfill$\square$\end{rem}

\begin{rem}[Estimation of Critical Values]\label{rem:crit:ex}
 For the estimation of critical values of Theorem  \ref{prop:ex} and \ref{prop:ex:tau}, let us define $\bold U_n^{\textsl{e}}=\big(Y_1-\osol_\k(Z_1),\dots,Y_n-\osol_\k(Z_n)\big)^t$. For any $m\geq 1$ we estimate the covariance matrix $\varSigma_m$ by 
\begin{equation*}
 \widehat\varSigma_m:=n^{-1}\,\bold W_m(\tau)^t\,\text{diag}(\bold U_n^{\textsl{e}})^2\,\bold W_m(\tau).
\end{equation*}
Now the asymptotic result of Theorem \ref{prop:ex} continues to hold if we replace $\varsigma_\m$ by the  Frobenius norm of $\widehat\varSigma_\m$ and $\mu_\m$ by the trace of $\widehat\varSigma_\m$. This consistency is shown in Lemma \ref{lem:crit}. In the setting of Theorem \ref{prop:ex:tau}, we replace $\varSigma^\textsl{e}$ by a finite dimensional matrix
\begin{equation*}
  \widehat\varSigma_M^{\textsl{e}}:=n^{-1}
\bold W_M(\tau)^t\Big(I_n-n^{-1}\bold Z_\k\bold Z_\k^t\Big)
\,\text{diag}(\bold U_n^{\textsl{e}})^2\,\Big(I_n-n^{-1}\bold Z_\k\bold Z_\k^t\Big)\bold W_M(\tau)
\end{equation*}
where $M>0$ is a sufficiently large integer. Let $(\widehat\lambda_j^{\textsl{e}})_{1\leq j\leq M_n}$ denote the ordered eigenvalues of $\widehat\varSigma_{M_n}^{\textsl{e}}$.
Hence, we approximate $\sum_{j=1}^\infty \lambda_j^\textsl{e}\chi_{1j}^2$ by the finite sum $\sum_{j=1}^{M_n}\widehat\lambda_j^{\textsl{e}}\,\chi_{1j}^2$ where $\max_{1\leq j\leq M_n}|\widehat\lambda_j^{\textsl{e}}-\lambda_j^{\textsl{e}}|=o_p(1)$  if $M_n=o(\sqrt n)$. 
\hfill$\square$\end{rem}
\begin{lem}\label{lem:crit}
Consider $\widehat\varSigma_\m$ as defined in Remark \ref{rem:crit:ex}. Under conditions of Theorem \ref{prop:ex} or Theorem \ref{prop:ex:tau} the difference of its Frobenius norm to $\varsigma_\m$ and its trace to $\mu_\m$  converge in probability to zero.
\end{lem}

\subsection{Limiting behavior under local alternatives and consistency.}
Similar to the previous sections we study the power and consistency properties of our test. 
Let us study the power of our test of exogeneity under linear local alternatives \eqref{loc:alt:ind} or \eqref{loc:alt:ind:2}. In these cases, it holds $\Ex[U|W]=0$ but $\Ex[U|Z]= -\varsigma_\m^{1/2}n^{-1/2} \delta(Z)$ under \eqref{loc:alt:ind} or $\Ex[U|Z]=-n^{-1/2}\delta(Z)$ under \eqref{loc:alt:ind:2}. 
\begin{prop}\label{coro:norm:ex}Given the conditions of Theorem \ref{prop:ex}  and Assumption  \ref{Ass:power} (ii) it holds under \eqref{loc:alt:ind}
\begin{equation*}
 (\sqrt 2\varsigma_\m)^{-1}\big(n\,S_n^{\textsl{e}}-\mu_\m\big)\stackrel{d}{\rightarrow}\cN\Big(2^{-1/2}\sum_{j= 1}^\infty\delta_j^2,1\Big).
\end{equation*}
Given the conditions of Theorem \ref{prop:ex:tau} it holds under \eqref{loc:alt:ind:2}
\begin{equation*}
n\, S_n^{\textsl{e}}\stackrel{d}{\rightarrow}\sum_{j=1}^\infty\lambda_j^\textsl{e}\,\chi_{1j}^2(\delta_j/\lambda_j^\textsl{e}).
\end{equation*}
\end{prop}

Let us now establish consistency of our tests when $H_\text e$ does not hold, that is, $\PP\big(\sol=\sol_0\big)<1$.
\begin{prop}\label{prop:ex:cons}Assume that $H_{\textsl{e}}$ does not hold. Let $\Ex|Y-\sol_0(Z)|^4<\infty$ and  Assumption  \ref{Ass:power} (i) hold true. Let $(\alpha_n)_{n\geq 1}$ as in Proposition \ref{prop:cons}. Under the conditions of Theorem \ref{prop:ex} we have
\begin{equation*}
 \PP\Big((\sqrt 2\,\varsigma_\m)^{-1}\big( nS_n^{\textsl{e}}-\mu_\m\big)>\alpha_n\Big)=1+o(1),
\end{equation*}
whereas in the setting of Theorem \ref{prop:ex:tau}
\begin{equation*}
  \PP\big(n\,S_n^{\textsl{e}}>\alpha_n\big)=1+o(1).
\end{equation*}
\end{prop}

In the following we show that our tests are consistent uniformly over the function class
\begin{equation*}
 \cI_{n}^\sr=\Big\{\sol\in \cL_Z^2:\,\|LT(\sol-\sol_0)\|_W^2\geq \sr\, n^{-1}\varsigma_\m\text{  and}\sup_{z\in\cZ}|(\sol-\sol_0)(z)|\leq C\Big\}
\end{equation*} 
form some constant $C>0$. 
For the next result let $q_{1\alpha}$ and $q_{2\alpha}$ denote the $1-\alpha$ quantile of $\cN(0,1)$ and $\sum_{j=1}^\infty\lambda_j^{\textsl{e}}\,\chi_{1j}^2$, respectively.
\begin{prop}\label{prop:ex:unf:cons} Let Assumption  \ref{Ass:power} be satisfied.
Under the conditions of Theorem \ref{prop:ex} we have for any $\varepsilon>0$, any  $0<\alpha<1$, and any sufficiently large constant $\rho>0$ that
\begin{equation*}
  \lim_{n\to\infty}\inf_{\sol\in\cI_n^\sr}\PP\Big((\sqrt 2\,\varsigma_\m)^{-1}\big(n\,S_n^{\textsl{e}}-\mu_\m\big)>q_{1\alpha}\Big)\geq 1-\varepsilon,\label{prop:ex:unf:cons:1}\\
\end{equation*}
whereas under the conditions of Theorem \ref{prop:ex:tau} it holds
\begin{equation*}
 \lim_{n\to\infty}\inf_{\sol\in\cI_{n}^\sr}\PP\big(n\,S_n^{\textsl{e}}>q_{2\alpha}\big)\geq 1-\varepsilon.\label{prop:ex:unf:cons:2}
\end{equation*}
\end{prop}

\section{A nonparametric specification test}
A solution to the linear operator equation \eqref{model:EE:t} only exists if $g$ belongs to the range of $\op$. 
This might be violated if, for instance, the instrument is not valid, that is, $\Ex[U|W]\neq 0$. In many economic applications \textit{a priori} smoothness restriction on the unknown function can be justified which we capture by a set of functions $\cF$.
We consider the hypothesis
 $H_{\textsl{np}}$: there exists a solution $\sol_0\in\cF$ to \eqref{model:EE:t}. 
The alternative hypothesis is that there exists a solution \eqref{model:EE:t} which does not belong to $\cF$. Under the alternative only unsmooth functions solve the conditional moment restriction which can be interpreted as a failure of validity of the instrument $W$.
We see in this section that our results allow also for a test of dimension reduction of the vector of regressors $Z$, that is, whether some regressors can be omitted from the structural function $\sol_0$.
\subsection{Nonparametric estimation method}
\paragraph{The nonparametric estimator.}
In the following, we  derive an estimator of $\sol_0$ under the null hypothesis $H_{\textsl{np}}$.
For simplicity, assume that $\cZ=\cW$ and consider a sequence $\{e_j\}_{j\geq 1}$ of approximating functions which are orthonormal on $\cZ$ with respect to the Lebesque measure $\nu$.
Under conditions given below, $\sol_0$ has the expansion $\sol_0(\cdot)=\sum_{l=1}^\infty \int\sol_0(z)e_l(z)\nu(z)dz\,e_l(\cdot)$.
Thereby, the conditional moment restriction under $H_{\textsl{np}}$ leads to the following unconditional moment restrictions 
\begin{equation}
\Ex[Y e_j(W)]=\sum_{l=1}^\infty \Ex[e_j(W)e_l(Z)]\int \sol_0(z)e_l(z)\nu(z)dz
\end{equation}
 for $j \geq 1$. This motivates the following  orthogonal series type estimator.
Let $\bold Z_k$ and $\bold Y_n$ be as in the previous section and let $\bold X_k$ denote a $n\times k$ matrix with entries $e_j(W_i)$ for $1\leq i\leq n$ and $1\leq j\leq k$.
Then for each $k\geq 1$ we consider the estimator
\begin{equation}\label{gen:def:est2}
\hsol_k(\cdot):=
e_\uk(\cdot)^t\widehat \beta_k\quad \text{where } \quad \widehat \beta_k=(\bold X_k^t\bold Z_k)^-\,\bold X_k^t \bold Y_n.
\end{equation}
Under conditions given below $\bold X_\k^t\bold Z_\k$ will be nonsingular with probability approaching one and hence its generalized inverse will be the standard inverse.
The nonparametric estimator $\hsol_k$ given in \eqref{gen:def:est2} was studied by \cite{Johannes2009}, \cite{Horowitz2011}, and \cite{Horowitz2009}.

\paragraph{Additional assumptions.}
In the following, we specify \textit{a priori} smoothness assumptions captured by the set $\cF$. 
As noted by \cite{Horowitz2009}, uniformly consistent testing of $H_{\textsl{np}}$ is only possible if the null is restricted that any solution to \eqref{model:EE:t} is smooth. Here, we assume that under the null hypothesis $\sol_0$ belongs to the ellipsoid
 $\cF:=\cF_{\sw}^\sr := \big\{\phi\in \cL_Z^2: \sum_{j= 1}^\infty\sw_j\Ex[\phi(Z)e_j(Z)]^2\leq \sr\big\}$. As in the previous section, $\sw=(\sw_j)_{j\geq 1}$ measures the approximation error of $\sol_0$ with respect to the basis $\{e_j\}_{j\geq 1}$. 

Further, as usual in the context of nonparametric instrumental regression, we specify some mapping properties of 
the conditional expectation operator $T$.
Denote by $\cT$ the set of all nonsingular operators mapping $\cL_Z^2$ to $\cL^2_W$. Given a sequence of weights $\Opw:=(\Opw_j)_{j\geqslant 1}$ and $\Opd\geqslant
1$ we define the subset $\Opwd$ of $\cT$ by
\begin{equation*}
\Opwd:=\Bigl\{ T\in\cT:\quad  \int|(\Op \phi)(w)|^2\nu(w)dw\leqslant {\Opd}\, \sum_{j= 1}^\infty\opw_j\Big(\int\phi(z)e_j(z)\nu(z)dz\Big)^2\quad \text{for all }\phi \in \cL^2_Z\Bigr\}.
\end{equation*}
If $p_Z/\nu$ is bounded from above and $p_W/\nu$ is uniformly bounded away from zero then the conditional expectation operator $T$ belongs to $\Opwd$ with $\opw_j=1$, $j\geq 1$, due to Jensen's inequality.
Notice that for all $T\in\Opwd$ it follows that  $\normV{T\basZ_j}_W^2\leq d\,\eta_p
\Opw_j$ and thereby, the condition $T\in\Opwd$ links the operator $T$ to the basis $\{e_j\}_{j\geq 1}$. In the following, we denote $[\op]_\uk=\Ex[e_\uk(W)e_\uk(Z)^t]$ which is assumed to be a nonsingular matrix.
In what follows, we introduce a stronger
  condition on the basis $\{e_l\}_{l\geq1}$.  We denote by $\cTdDw$ for some $\tD\geq \td$ the subset of
  $\cTdw $ given by
\begin{equation*}
\cTdDw:=\Bigl\{ T\in \cTdw :[\op]_\uk \text{ is nonsingular and}\quad \sup_{k\geq 1}\normV{\text{diag}(\opw_1,\dots,\opw_k)^{1/2}[\op]_\uk^{-1}}^2\leqslant \tD\Bigr\}.
\end{equation*}
The class $\cTdDw$ only contains operators $T$ whose off-diagonal elements of $[T]_\uk^{-1}$ are sufficiently small for all $k\geq 1$. A similar diagonality restriction has been used by \cite{HallHorowitz2005} or \cite{BJ2011}.
Besides the mapping properties for the operator $T$ we need a stronger assumption for the basis under consideration.
The following condition gathers conditions on the sequences $\sw$ and $\opw$.
\begin{assA}\label{Ass:A6} 
(i) Under $H_{\textsl{np}}$, let $\sol_0\in\cF_\sw^\sr$ with nondecreasing sequence $\sw$ satisfying $j^3=o(\sw_j)$. (ii) The sequence $\{e_j\}_{j\geq 1}$ is an  orthogonal basis on $\cZ=\cW$ with respect to $\nu$.
 (iii) There exists some constant $\eta_e\geqslant 1$ such that  $\sup_{j\geq 1}\sup_{z\in\cZ}|e_j(z)|\leqslant \eta_e$.
(iv) Let $T\in\cTdDw$ with  $\Opw$ being a strictly positive sequences  such that $\opw$ and $(\opw_{j}/\tau_{j})_{j\geq 1}$ are nonincreasing. (v) $p_Z/\nu$ is bounded from above and $p_W/\nu$ is uniformly bounded away from zero.
\end{assA}
Due to Assumption \ref{Ass:A6} $(iv)$ the degree of additional smoothing for our testing procedure must not be stronger than the degree of ill-posedness implied by the conditional expectation operator $T$.
Under similar assumptions as above, \cite{Johannes2009} show that mean integrated squared error loss of $\hsol_{\k}$
attains the optimal rate of convergence $\cR_n:=\max\big(\sw_{\k}^{-1}, \sum_{j=1}^{\k}(n\opw_j)^{-1}\big)$. Due to Assumption \ref{Ass:A6} $(v)$ we do not require orthonormal bases with respect to the unknown distribution $(Z,W)$ (cf. Remark 3.2 of \cite{BJ2011}).

\subsection{The test statistic and its asymptotic distribution}
As in the previous sections, our test is based on the observation that the null hypothesis $H_{\textsl{np}}$ is equivalent to $L(g-T\sol_0)=0$. 
Our goodness-of-fit statistic for testing nonparametric specifications is given by $S_n$ where $\sol_0$ is replaced by the nonparametric estimator  $\hsol_{\k}$ given in \eqref{gen:def:est2}, that is, 
\begin{equation*}
  S_n^{ \textsl{np}}:=\sum_{j=1}^\m\tau_j\big|n^{-1}\sum_{i=1}^n\big(Y_i-\hsol_\k(Z_i)\big)f_j(W_i)\big|^2.
\end{equation*}
If $S_n^{\textsl{np}}$ becomes too large then there exists no function in $\cF_\sw^\sr$ solving \eqref{model:EE:t}.
The next result establishes asymptotic normality of $S_n^{\textsl{np}}$ after standardization. Again, a key requirement to obtain this asymptotic distribution is that $\k=o(\varsigma_\m)$ implying that $\k=o(\m)$ if the smoothing operator $L$ is the identity. This corresponds to the test of overidentification in the parametric framework where more orthogonality restrictions than parameters are required.
\begin{theo}\label{prop:np} Let Assumptions \ref{Ass:A1}--\ref{Ass:A4} and \ref{Ass:A6} be satisfied. If
\begin{equation}\label{cond:prop:np}
  n\opw_\k=o(\sw_\k\varsigma_{\m}), \,\k=o(\varsigma_{\m}),\,\k\Big(\sum_{j=1}^\m\tau_j\Big)^2=O(n\opw_\k),\,\text{and } \Big(\sum_{j=1}^\m\tau_j\Big)^3=o(n)
\end{equation}
then it holds under $H_{\textsl{np}}$
\begin{equation*}
 (\sqrt 2\varsigma_{\m})^{-1} \Big(nS_n^{\textsl{np}}-\mu_{\m}\Big)\stackrel{d}{\rightarrow}\mathcal N(0,1).
\end{equation*}
\end{theo}
\begin{example}
Consider the setting of Example \ref{exmp:cases}. In the \textit{mildly ill posed case} where $\opw_j\sim j^{-2a/r}$ for some $a\geq0$ condition  \eqref{cond:prop:np} holds true if $\k\sim n^\kappa$ with $\kappa<\nu/2$ and 
\begin{equation*}
 {r(1-\nu/2)/(2a+2p)}<\kappa<r(1-2\nu)/(2a+r).
\end{equation*}
In the \textit{severely ill posed case}, that is, $\opw_j\sim \exp(-j^{2a/r})$ for some  $a>0$, condition  \eqref{cond:prop:np} is satisfied if, for example, $\m$ satisfies $\m=o(k_n^p)$ and $k_n=o(\sqrt\m)$ where $\k\sim\big(\log n-\log (m_n^{3/2})\big)^{r/(2a)}$.
$\hfill\square$
\end{example}

The next result states an asymptotic distribution of our test if $\sum_{j=1}^\m\tau_j=O(1)$.
Let  $\varSigma^\textsl{np}$ be the covariance matrix of the infinite dimensional centered vector  $\big(U(f_j^\tau(W)-e_j^\tau(W))\big)_{j\geq 1}$.
The ordered eigenvalues of $\varSigma^\textsl{np}$ are denoted by $(\lambda_j^\textsl{np})_{j\geq1}$.

\begin{theo}\label{prop:np:tau} Let Assumptions \ref{Ass:A1}--\ref{Ass:A4} and \ref{Ass:A6} be satisfied. If
\begin{equation}\label{cond:prop:np:tau}
  \sum_{j=1}^\m\tau_j=O(1), \quad n\opw_\k=o(\sw_\k),\quad k_n^3=o(n\opw_\k), \text{ and }m_n^{-1}=o(1)
\end{equation}
then it holds under $H_{\textsl{np}}$
\begin{equation*}
n\,S_n^{\textsl{np}}\stackrel{d}{\rightarrow}
\sum_{j=1}^\infty\lambda_j^\textsl{np}\,\chi_{1j}^2.
\end{equation*}
\end{theo}

\begin{example}
Consider the setting of Example \ref{exmp:cases:2}. In the \textit{mildly ill posed case}, that is, $\opw_j\sim j^{-2a/r}$ for some $a\geq0$,
condition \eqref{cond:prop:np:tau} is satisfied if $\m\sim n^\nu$ for some $\nu>0$ and $\k\sim n^\kappa$ with
\begin{equation*}
 r/(2a+2p)<\kappa<r/(2a+3r).
\end{equation*}
In the \textit{severely ill posed case}, that is, $\opw_j\sim \exp(-j^{2a/r})$ for some $a>0$, condition \eqref{cond:prop:np:tau}
is satisfied if $\k\sim\big(\log (n^{1+\varepsilon})\big)^{r/(2a)}$ for any $\varepsilon>0$. In contrast to Theorem \ref{prop:np}, we require undersmoothing of the estimator $\hsol_\k$.
$\hfill\square$
\end{example}
\begin{rem}
 If the basis $\{e_j\}_{j\geq 1}$ coincides with $\{f_j\}_{j\geq 1}$ then $n\,S_n^{\textsl{np}}$ is asymptotically degenerate. 
To avoid this degeneracy problem we choose different bases functions and hence, sample splitting 
as used by \cite{Horowitz2009} is not necessary here. 
$\hfill\square$
\end{rem}

\begin{rem}Let $Z'$ be a vector containing only entries of $Z$ with $\dim(Z')<\dim(Z)$.
 It is easy to generalize our previous result for a test of
$H_{\textsl{np}}'$: there exists a solution $\sol_0\in\cF_\sw^\sr$ to \eqref{model:EE:t} only depending on $Z'$.
To be more precise consider the test statistic
\begin{equation*}
  S_n^{'\textsl{np}}:=\big\|n^{-1}\sum_{i=1}^n\big(Y_i-\hsol_\k(Z_i')\big)f_\umn^\tau(W_i)\|^2
\end{equation*}
where $\hsol_\k$ is the estimator \eqref{gen:def:est2} based on an iid. sample $(Y_1,Z_1',W_1),\dotsc,(Y_n,Z_n',W_n)$ of $(Y,Z',W)$. Under $H_{\textsl{np}}'$ we consider the conditional expectation operator $T':\cL_{Z'}^2\to \cL_W^2$ with $T'\phi:=\Ex[\phi(Z')|W]$. It is interesting to note that if $T$ is nonsingular then also $T'$ is. Hence, for a test of $H_{\textsl{np}}'$  we may replace Assumption \ref{Ass:A3} by the weaker condition that $T'$ is nonsingular. Moreover, under $H_{\textsl{np}}'$ the results of Theorem \ref{prop:np} and \ref{prop:np:tau} still hold true if we replace $Z$ by $Z'$.
$\hfill\square$
\end{rem}
In the mildly ill-posed case, the estimation precision suffers from the curse of dimensionality. 
Hence, by the  test of dimension reduction of $Z$ we can increase the accuracy of estimation of $\sol_0$. On the other hand, in the severely ill-posed case the rate of convergence is independent of the dimension of $Z$ (cf. \cite{ChenReiss2008}).
As the next example illustrates, a dimension reduction test can also weaken the required restrictions on the instrument to obtain identification of $\sol$ in the restricted model.
\begin{example}
 Let $Z=(Z^{(1)},Z^{(2)})$ where both, $Z^{(1)}$ and $Z^{(2)}$ are endogenous vectors of regressors. But only $Z^{(1)}$ satisfies a sufficiently strong relationship with the instrument $W$ in the sense that for all $\phi\in \cL_{Z^{(1)}}^2$ condition $\Ex[\phi(Z^{(1)})|W]=0$ implies $\phi=0$. In this example, we do not assume that this completeness condition is fulfilled for the joint distribution of $(Z^{(2)},W)$. 
Thereby only the operator $T^{(1)}:\cL_{Z^{(1)}}^2\to \cL_W^2$ with $T^{(1)}\phi:=\Ex[\phi(Z^{(1)})|W]$ is nonsingular but $T$ is singular. If our dimension reduction test of $Z$ indicates that  $Z^{(2)}$ can be omitted from the structural function $\sol_0$ then we obtain identification in the restricted model.
$\hfill\square$
\end{example}

\begin{rem}\label{rem:crit:np}[Estimation of Critical Values]
 For the estimation of critical values of Theorem  \ref{prop:np} and \ref{prop:np:tau}, let us define $\bold U_n^{\textsl{np}}=\big(Y_1-\hsol_\k(Z_1),\dots,Y_n-\hsol_\k(Z_n)\big)^t$. For all $m\geq 1$, we estimate the covariance matrix $\varSigma_m$ by 
\begin{equation*}
 \widehat\varSigma_m:=n^{-1}\,\bold W_m(\tau)^t\text{diag}(\bold U_n^{\textsl{np}})^2\,\bold W_m(\tau).
\end{equation*}
Now the asymptotic result of Theorem \ref{prop:np} continues to hold if we replace $\varsigma_\m$ by the  Frobenius norm of $\widehat\varSigma_\m$ and $\mu_\m$ by the trace of $\widehat\varSigma_\m$ (this is easily seen from the proof of Lemma \ref{lem:crit} assuming that $\{f_j\}_{j\geq 1}$ is uniformly bounded). In the setting of Theorem \ref{prop:np:tau}, we replace $\varSigma^\textsl{np}$ by a finite dimensional matrix. Let $\bold V_{k}:=\bold W_k \big(\bold Z_k^t\bold W_k)^{-1}\bold Z_k^t$ for $k\geq 1$. Then for a sufficiently large integer $M>0$ we estimate $\varSigma^\textsl{np}$ by 
\begin{equation*}
  \widehat\varSigma_M^{\textsl{np}}:=n^{-1}
\bold W_M(\tau)\big(I_n-\bold V_{\k}\big)^t\,
\text{diag}(\bold U_n^{\textsl{np}})^2\,\big(I_n-\bold V_{\k}\big)\bold W_M(\tau).
\end{equation*}
Hence, we approximate $\sum_{j=1}^\infty \lambda_j^{\textsl{np}}\chi_{1j}^2$ by the finite sum $\sum_{j=1}^{M_n}\widehat\lambda_j^{\textsl{np}}\,\chi_{1j}^2$ where $(\widehat\lambda_j^{\textsl{np}})_{1\leq j\leq M_n}$ are the ordered eigenvalues of $\widehat\varSigma_{M_n}^{\textsl{np}}$ where $\max_{1\leq j\leq M_n}|\widehat\lambda_j^{\textsl{np}}-\lambda_j^{\textsl{np}}|=o_p(1)$ if $M_n=o(\sqrt n)$.
\hfill$\square$\end{rem}

\subsection{Limiting behavior under local alternatives and consistency.}
Similar to the previous sections we study the power and consistency properties of our test.
To study the power against local alternatives of the statistic $S_n^{\textsl{np}}$ we consider alternative models with the function $\sol_\k(\cdot)=e_\k(\cdot)^t[T]_\k^{-1}\Ex[Yf_\ukn(W)]$. 
We consider alternative models
\begin{equation}\label{loc:alt:np}
Y=\sol_\k(Z)+\varsigma_\m^{1/2}n^{-1/2}\delta(Z)+U
\end{equation}
for some function $\delta\in L_Z^4$ and where $\Ex[U|W]=0$. Let $\sol$ be such that $\Ex[Y-\sol(Z)|W]=0$. Due to \eqref{loc:alt:np} $\sol$ does not belong $\cF_\sw^\rho$ and hence $H_{\textsl{np}}$ fails.
Indeed, if $\sol\in \cF_\sw^\rho$ then we show in the appendix that $\|T(\sol-\sol_\k)\|_W^2=O(\opw_\k\sw_\k^{-1})=o(\varsigma_\m n^{-1})$ due to condition \eqref{cond:prop:np} (or \eqref{cond:prop:np:tau}), which is in contrast to \eqref{loc:alt:np}. 

\begin{prop}\label{coro:norm:np}Let  Assumption    \ref{Ass:power} (ii) be satisfied. Given the conditions of Theorem \ref{prop:np} it holds under  \eqref{loc:alt:np}
\begin{equation*}
 (\sqrt 2\varsigma_\m)^{-1}\big(n\,S_n^{\textsl{np}}-\mu_\m\big)\stackrel{d}{\rightarrow}\cN\Big(2^{-1/2}\sum_{j= 1}^\infty\xi _j^2,1\Big).
\end{equation*}
Given the conditions of Theorem \ref{prop:np:tau} it holds under \eqref{loc:alt:np} where $\varsigma_\m$ is replaced by $1$ that
\begin{equation*}
 n\, S_n^{\textsl{np}}\stackrel{d}{\rightarrow}\sum_{j=1}^\infty\lambda_j^\textsl{np}\,\chi_{1j}^2(\delta_j/\lambda_j^\textsl{np}).
\end{equation*}
\end{prop}

In the next proposition, we establish consistency of our test when $H_{\textsl{np}}$ does not hold, that is, the solution to \eqref{model:EE:t} does not belong to  $\cF_\sw^\sr$ for any sequence $\sw$ satisfying Assumption \ref{Ass:A6} and any sufficiently large constant $0<\sr<\infty$.
\begin{prop}\label{prop:np:cons}Assume that $H_{\textsl{np}}$ does not hold. Let  Assumption  \ref{Ass:power} (i) hold true. Let $(\alpha_n)_{n\geq 1}$ be as in Proposition \ref{prop:cons}. Under the conditions of Theorem \ref{prop:np} and \ref{prop:np:tau}, respectively, we have
\begin{gather*}
 \PP\Big((\sqrt 2\,\varsigma_\m)^{-1}\big( nS_n^{\textsl{np}}-\mu_\m\big)>\alpha_n\Big)=1+o(1),\\
  \PP\big(nS_n^{\textsl{np}}>\alpha_n\big)=1+o(1).
\end{gather*}
\end{prop}

In the following we show that our tests are consistent uniformly over the function class
\begin{equation*}
 \cJ_{n}^\sr=\Big\{\sol\in \cL_Z^2:\|LT(\sol-\sol_0)\|_W^2\geq \sr n^{-1}\varsigma_\m\text{  and}\sup_{z\in\cZ}|(\sol-\sol_0)(z)|\leq C\Big\}
\end{equation*} 
where $\sol_0\in\cF_\sw^\rho$ solves \eqref{model:EE:t} and  $C>0$ is a finite constant.
 For the next result let $q_{1\alpha}$ and $q_{2\alpha}$ denote the $1-\alpha$ quantile of $\cN(0,1)$ and $\sum_{j=1}^\infty\lambda_j^{\textsl{np}}\,\chi_{1j}^2$, respectively.
\begin{prop}\label{prop:np:unf:cons}Let Assumption  \ref{Ass:power} be satisfied.
For any $\varepsilon>0$, any  $0<\alpha<1$, and any sufficiently large constant $\rho>0$ we have under the conditions of Theorem \ref{prop:np}
\begin{equation*}
  \lim_{n\to\infty}\inf_{\sol\in\cJ_n^\sr}\PP\Big((\sqrt 2\,\varsigma_\m)^{-1}\big(n\,S_n^{\textsl{np}}-\mu_\m\big)>q_{1\alpha}\Big)\geq 1-\varepsilon,\label{prop:np:unf:cons:1}\\
\end{equation*}
whereas under the conditions of Theorem \ref{prop:np:tau}  it holds
\begin{equation*}
\lim_{n\to\infty}\inf_{\sol\in\cJ_{n}^\sr}\PP\big(n\,S_n^{\textsl{np}}>q_{2\alpha}\big)\geq 1-\varepsilon.\label{prop:np:unf:cons:2} 
\end{equation*}
\end{prop}
\section{Monte Carlo simulation}
In this section, we study the finite-sample performance of our test by presenting the results of Monte Carlo experiments. There are $1000 $ Monte Carlo replications in each experiment. Results are presented for the nominal level $0.05$.
 Realizations of Y were
generated from
\begin{equation}\label{sim:mod}
Y = \sol(Z) + c_U U
\end{equation}
for some constant $c_U>0$ specified below.
 The structural function $\sol$ and the joint distribution of $(Z,W,U)$ varies in the experiments below. As basis $\{f_j\}_{j\geq 1}$ we choose cosine basis functions given by $f_{j}(t)=\sqrt{2}\cos(\pi j t)$ for $j=1,2,\dots$ throughout this simulation study.
\paragraph{Parametric Specification}
Let us investigate the finite sample performance of our tests in the case of parametric specifications. 
 Realizations $(Z,W)$ were generated by $ W \sim \cU[0,1]$, $Z = (\xi\, W +(1-\xi)\, \varepsilon)^2$ where $\xi=0.8$ and $\varepsilon\sim\cN(0.5, 0.1)$. 
Moreover, let $U = \kappa\,\varepsilon + \sqrt{1-\kappa^2}\,\epsilon$ with $\kappa=0.3$ and $\epsilon\sim N(0, 1)$.
Then realizations of $Y$ where generated by \eqref{sim:mod} with $c_U=0.2$ by an either linear function
\begin{equation}\label{mod:sim:1}
\sol(z)= z, 
\end{equation}
a polynomial of second degree
\begin{equation}\label{mod:sim:2}
\sol(z)= z- z^2,
\end{equation}
or a polynomial of third degree
\begin{equation}\label{mod:sim:3}
\sol(z)=z- z^2+\theta_3\, z^3.
\end{equation} 
Given \eqref{mod:sim:3} is the correct model, then $\theta_3=1.5$ if \eqref{mod:sim:1} is the null model and $\theta_3=3$ if \eqref{mod:sim:1} is the null model.
In Table 1 we depict the empirical rejection probabilities when using $S_n^\textsl{p}$ with additional smoothing where either $\tau_j= j^{-1}$ or $\tau_j= j^{-2}$, $j\geq 1$,  which we denote by $S_n^{1\textsl{p}}$ or $S_n^{2\textsl{p}}$, respectively.
\begin{table}[h!]
\begin{center}
\renewcommand{\arraystretch}{1.1}
\begin{tabular}{llllr}
\textit{Sample }&		\,\,\textit{Null}& \,\,\,\textit{Alt.} &\textit{\quad\quad Empirical Rejection probability }\\
  \,\,\,\textit{Size } & \textit{Model }  & 	\textit{Model}&\quad\quad \textit{\,\,$S_n^{1\textsl{p}}$\quad\quad\quad\quad$S_n^{2\textsl{p}}$}\textit{\quad\quad\,\,\,\,H(2006)' test}\\
 \hline	\hline
\quad 250  &\quad\eqref{mod:sim:1}	&  {\small $H^{\textsl p}$ true}& \quad\quad 0.047 \quad\quad\quad 0.045 \quad\quad\quad0.063 \\
	    &\quad\eqref{mod:sim:2}	&   {\small $H^{\textsl p}$ true} & \quad\quad 0.049 \quad\quad\quad0.050 \quad\quad\quad0.059\\
	    &\quad\eqref{mod:sim:1}	& \quad\eqref{mod:sim:2} &\quad\quad 0.902 \quad\quad\quad0.930 \quad\quad\quad0.888  \\
	    &\quad\eqref{mod:sim:1}	&\quad \eqref{mod:sim:3} &\quad\quad 0.730 \quad\quad\quad0.732 \quad\quad\quad0.653 \\
	    &\quad\eqref{mod:sim:2}	&\quad \eqref{mod:sim:3} & \quad\quad 0.442 \quad\quad\quad0.488 \quad\quad\quad0.468\\
\hline
\quad 500  &\quad\eqref{mod:sim:1}	&   {\small $H^{\textsl p}$ true} & \quad\quad0.055 \quad\quad\quad0.044 \quad\quad\quad0.053 \\
	    &\quad\eqref{mod:sim:2}	&   {\small $H^{\textsl p}$ true} & \quad\quad0.051 \quad\quad\quad0.053 \quad\quad\quad0.059\\
	    &\quad\eqref{mod:sim:1}	& \quad\eqref{mod:sim:2} &\quad\quad 0.989 \quad\quad\quad0.998 \quad\quad\quad0.988  \\
	    &\quad\eqref{mod:sim:1}	&\quad \eqref{mod:sim:3} &\quad\quad 0.899 \quad\quad\quad0.894 \quad\quad\quad0.780 \\
	    &\quad\eqref{mod:sim:2}	&\quad \eqref{mod:sim:3} & \quad\quad0.709 \quad\quad\quad0.728 \quad\quad\quad0.652\\
\hline
\end{tabular}
\end{center}
\caption{Empirical Rejection probabilities for parametric specification}
\end{table}
When $\tau_j=j^{-1}$ then the number of basis functions used is $m=200$ while in the case of $\tau_j=j^{-2}$ a choice of $m=100$ is sufficient.
The critical values are estimated as described in Remark \ref{est:crit:par} where $M=150$ if $\tau_j=j^{-1}$ and $M=100$ if $\tau_j=j^{-2}$. This choice of $M$ ensures that the estimated eigenvalues $\widehat\lambda_j$ are sufficiently close to zero for all $j\geq M$.
We compare our test statistic with the test of \cite{Horowitz2006}. 
We follow his implementation using biweight kernels. The bandwidth used to estimate the joint density of $(Z,W)$ was also selected by cross validation.
As Table 1 illustrates, the results for $S_n^{1\textsl{p}}$ and $S_n^{2\textsl{p}}$ are quite similar. In both situations, our test is more powerful than the test of \cite{Horowitz2006} when testing \eqref{mod:sim:1} against \eqref{mod:sim:3}. 
In this simulation study, we observed that the estimated coefficients of $T(\sol-\phi(\thet,\cdot))$ have a fast decay. Consequently, the test statistic $S_n$ with no weighting has less power, as we discussed in Subsection 2.4. In contrast, we will demonstrate by the end of this section that using weights can be inappropriate. 
\paragraph{Testing Exogeneity}
We now turn to the test of exogeneity where the realizations $(Z,W)$ are generated by $W \sim\cU[0,1]$ and $Z = \xi\,W +\sqrt{1-\xi^2}\,\varepsilon$ with $\xi=0.7$,  and $\varepsilon\sim \cU[0,1]$.
 Moreover, let $U = \kappa\,\varepsilon + \sqrt{1-\kappa^2}\,\epsilon$ with $\epsilon\sim \cU[0,1]$. Here, $\kappa$ measures the degree of endogeneity of $Z$ and is varied among the experiments. The null hypothesis $H_0$ holds true if $\kappa=0$ and is false otherwise. Now realizations of $Y$ where generated by \eqref{sim:mod} with $c_U=1$ and the nonparametric structural function $\sol_1(z)=\sum_{j=1}^\infty (-1)^{j+1}\,j^{-1} \sin(j\pi z)$. For computational reasons we truncate the infinite sum at $K=100$. The resulting function is displayed in Figure \ref{phi-12}. We estimate the structural relationship using Lagrange polynomials. Indeed, only a few basis functions are necessary to accurately approximate the true function. If we choose $\k$ too small or too large then the estimator will be a poor approximate of the true structural function and hence, the test statistic will reject $H_{\textsl{np}}$. In this experiment we set $\k=4$ for $n=250$ and $n=500$.

%
\begin{table}[h]\label{table:ex}	
\begin{center}
\renewcommand{\arraystretch}{1.1}
\begin{tabular}{lllr}
\textit{Sample Size} &\,\,\textit{$\kappa$ }&\textit{\quad\quad\quad Empirical Rejection probability using}\\
      & 		\quad& \quad\textit{\,\,$S_n^{1\textsl{e}}$\quad\quad\quad\quad$S_n^{2\textsl{e}}$\quad\quad BH(2007)' test }\\
\hline	\hline

 \quad \quad250 &$0.0$	 &\quad 0.038 \quad\quad\quad0.030 \quad\quad\quad\quad0.030\\
	      &  $0.15$ &\quad 0.209 \quad\quad\quad0.314 \quad\quad\quad\quad0.153\\
	      &  $0.2$  &\quad 0.369 \quad\quad\quad 0.513 \quad\quad\quad\quad0.293\\
	      &  $0.25$ & \quad0.591 \quad\quad\quad0.716 \quad\quad\quad\quad0.504\\
\hline
\quad \quad500 &$0.0$	 &\quad 0.043 \quad\quad\quad0.043 \quad\quad\quad\quad0.052\\
	      &  $0.15$ &\quad 0.476 \quad\quad\quad0.543 \quad\quad\quad\quad0.416\\
	      &  $0.2$  &\quad 0.749 \quad\quad\quad0.809 \quad\quad\quad\quad 0.693\\
	      &  $0.25$ &\quad 0.922 \quad\quad\quad 0.957 \quad\quad\quad\quad0.885\\
\hline
\end{tabular}
\end{center}
\caption{Empirical Rejection probabilities for testing exogeneity}
\end{table}
In Table 2 we depict the empirical rejection probabilities when using $S_n^\textsl{e}$ with additional smoothing where either ${\tau_j} =j^{-1}$ or ${\tau_j}=j^{-2}$, $j\geq 1$, which we denote by $S_n^{1\textsl{e}}$ or $S_n^{2\textsl{e}}$, respectively. The critical values of these statistics are estimated as described in Remark \ref{rem:crit:ex} with $M=50$ in case of $\tau_j=j^{-1}$ and $M=40$ in case of $\tau_j=j^{-2}$.
 We compare our results with the test of \cite{Blundell2007}. We follow their approach by choosing the bandwidth of  the joint density of $(Z,W)$ by cross validation. The bandwidth of the marginal of $Z$ is $n^{1/5-7/24}$ times the cross-validation bandwidth. As we see from Table 2, $S_n^{1\textsl{e}}$ is slightly more powerful than the test of \cite{Blundell2007}. If we choose a stronger sequence, however,  then our test statistic $S_n^{2\textsl{e}}$ becomes considerably more powerful.
 \paragraph{Nonparametric Specification}
Let us now study the finite sample of our test in the case of nonparametric specification. We generate the pair $(Z,W)$ as in the parametric case described above. For the generation of the dependent variable $Y$ we distinguish two cases. 
Besides the structural function $\sol_1(z)=\sum_{j=1}^\infty (-1)^{j+1}j^{-2} \sin(j\pi z)$ we also consider the function $\sol_2(z)=\sum_{j=1}^\infty ((-1)^{j+1}+1)/4\,j^{-2} \sin(j\pi z)$. Again, for computational reasons we truncate the infinite sum at $K=100$. The resulting functions are displayed in Figure \ref{phi-12}.
Further, $Y$ is generated by \eqref{sim:mod} either with $\sol_1$ and $c_U=0.2$ or $\sol_2$ and $c_U=0.8$.  In both cases, we estimate the structural relationship using Lagrange polynomials with $\k=4$ for $n=500$ and $n=1000$.

If  $H_{\textsl{np}}$ is false then $\Ex[U|W]\neq 0$ and we let $\Ex[U|W]=\Ex[\rho(Z)|W]$ where $\rho$ is defined below.
Consequently, when $H_{\textsl{np}}$ is false we generate realizations of $Y$ from
\begin{equation*}
 Y=\sol_l(Z)+\rho_j(Z)+c_U U
\end{equation*}
 for $l=1,2$ and $j\geq 1$ where $\rho_j(z)=c_j(\exp(2jz)\1_{\{z\leq 1/2\}}+\exp(2j(1-z))\1_{\{z>1/2\}}-1)$ and $c_j$ is a normalizing constant such that $\int_0^1 \rho_j(z)dz=0.5$. The functions $\rho_j$ are continuous but not differentiable at $0.5$. Roughly speaking,  the degree of roughness of $\rho_j$ is larger for larger $j$.
\begin{figure}[h]\label{phi-12}
	\centering
		\includegraphics[width=12cm, height=6.7cm]{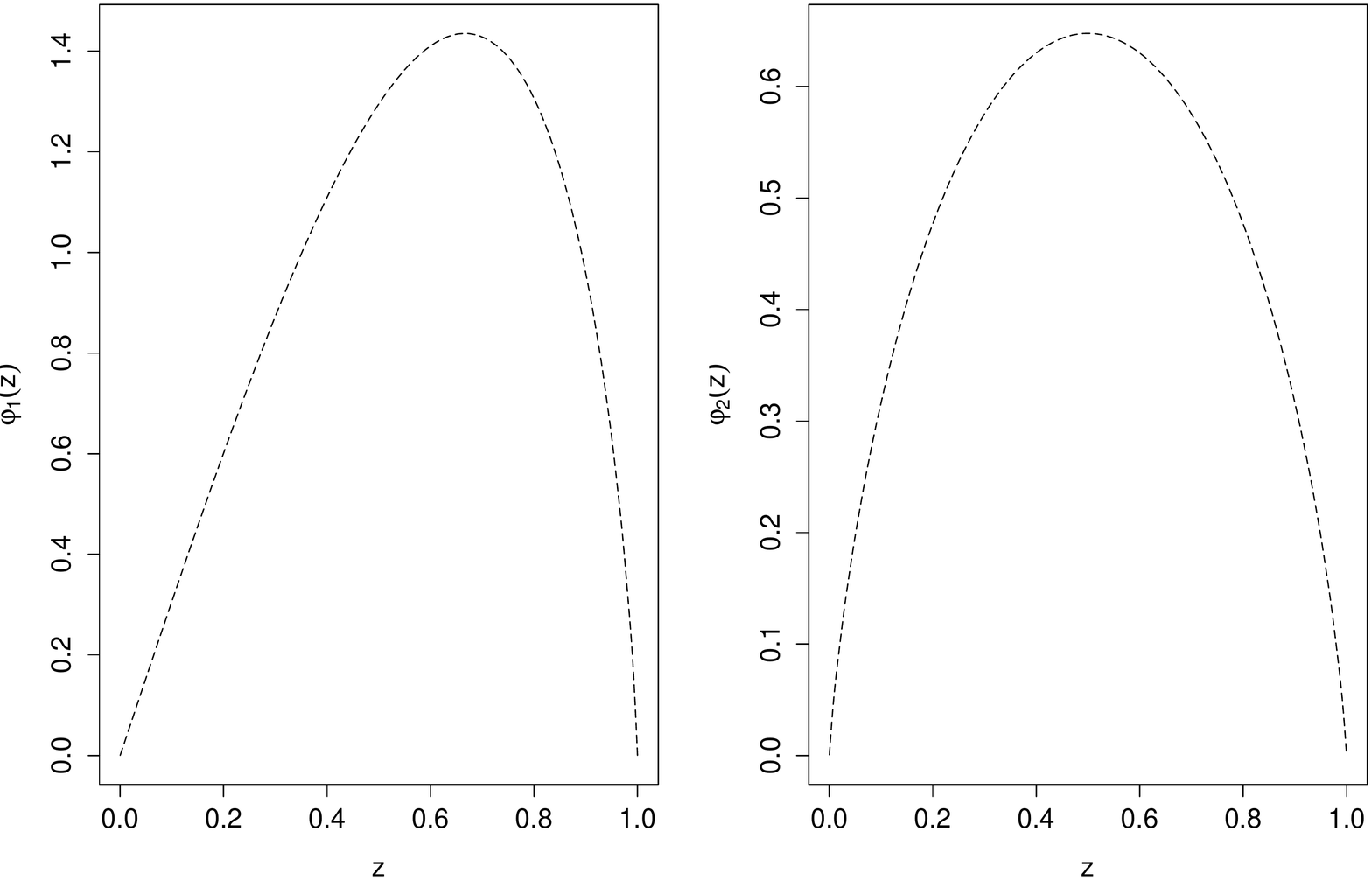}
	\caption{Graph of $\sol_1$ and $\sol_2$}
\end{figure}
\begin{table}[h]	
\begin{center}
\renewcommand{\arraystretch}{1.1}
\begin{tabular}{lllr}
\textit{Sample Size} &\,\,\textit{$\rho$ }&\textit{\quad\quad\quad Empirical Rejection probability using}\\
      & 		\quad& \quad\textit{\,\,$S_n^{0\textsl{np}}$\quad\quad\quad\quad$S_n^{2\textsl{np}}$\quad\quad H(2012)' test }\\
\hline	\hline
 \quad \quad500 &{\small $H^{\textsl{np}}$ true}&\quad 0.034 \quad\quad\quad 0.039 \quad\quad\quad\quad0.040\\
	      &  $\rho_1$ &\quad 0.099 \quad\quad\quad 0.382 \quad\quad\quad\quad0.258\\
	      &  $\rho_2$ &\quad 0.309  \quad\quad\quad0.765 \quad\quad\quad\quad0.536\\
 	      &  $\rho_4$  &\quad0.498 \quad\quad\quad 0.884 \quad\quad\quad\quad0.712\\
\hline
 \quad \quad1000 &{\small $H^{\textsl{np}}$ true}&\quad 0.058 \quad\quad\quad 0.058 \quad\quad\quad\quad0.046\\
	      &  $\rho_1$ &\quad 0.405 \quad\quad\quad 0.672 \quad\quad\quad\quad0.427\\
	      &  $\rho_2$ &\quad 0.768  \quad\quad\quad0.899 \quad\quad\quad\quad0.704\\
 	      &  $\rho_4$  &\quad0.920 \quad\quad\quad0.943 \quad\quad\quad\quad0.808\\
\hline
\end{tabular}
\end{center}
\caption{Empirical Rejection prob. for Nonparametric Specification for $\sol_1$ with $c_U=0.2$}
\end{table}
In Table 3, we depict the empirical rejection probabilities when using $S_n^\textsl{np}$ with either no smoothing or additional smoothing ${\tau_j}=j^{-2}$, $j\geq 1$, which we denote by $S_n^{0\textsl{np}}$ or $S_n^{2\textsl{np}}$, respectively.
When no additional smoothing is applied then the number of basis functions $f_j$ is given by $\m=11$ if $n=500$ and $\m=15$ if $n=1000$ and hence, the choice of $\m$ is slightly larger than $n^{1/3}$ as suggested by the theoretical results. The critical values of these statistics are estimated as described in Remark \ref{rem:crit:np} where in the case of $S_n^{2\textsl{np}}$ we choose $M=100$. We compare our results with the test of \cite{Horowitz2009}. We observe that the statistic $S_n^{0\textsl{np}}$ is less powerful than  $S_n^{2\textsl{np}}$ against the alternatives $\rho_1$ and $\rho_2$. 

In the following, we illustrate that using additional weighting can be inappropriate. Table 4 illustrates the power of our tests when the structural function $\sol_2$ is considered and realizations $(Z,W)$ were generated by $ W \sim \cU[0,1]$, $Z = (0.8\, W +0.3\, \varepsilon)^2$ where $\varepsilon\sim\cN(0.5, 0.05)$. In this case, we generate $Y$ using \eqref{sim:mod} where  $c_U=0.8$. 
In this case, the estimates of the generalized coefficients of $T(\sol-\sol_0)$ are more fluctuating and using weights is not appropriate here. Indeed, as we can see from Table 4, the test statistic $S_n^{0\textsl{np}}$ with no smoothing is more powerful than $S_n^{2\textsl{np}}$ were weighting $\tau_j=j^{-2}$, $j\geq 1$, is used. In particular, $S_n^{0\textsl{np}}$ is much more powerful than the test of \cite{Horowitz2009}.
\begin{table}[h]	
\begin{center}
\renewcommand{\arraystretch}{1.1}
\begin{tabular}{lllr}
\textit{Sample Size} &\,\,\textit{$\rho$ }&\textit{\quad\quad\quad Empirical Rejection probability using}\\
      & 		\quad& \quad\textit{\,\,$S_n^{0\textsl{np}}$\quad\quad\quad\quad$S_n^{2\textsl{np}}$\quad\quad H(2012)' test }\\
\hline	\hline
 \quad \quad500 &{\small $H^{\textsl{np}}$ true}&\quad 0.022 \quad\quad\quad 0.044 \quad\quad\quad\quad 0.044\\
				    &  $\rho_3$  &\quad 0.230 \quad\quad\quad 0.193 \quad\quad\quad\quad0.158\\
				    &  $\rho_4$  &\quad 0.400 \quad\quad\quad 0.319 \quad\quad\quad\quad0.245\\
				    &  $\rho_5$  &\quad 0.543 \quad\quad\quad 0.463 \quad\quad\quad\quad0.370\\
\hline
 \quad \quad1000 &{\small $H^{\textsl{np}}$ true}&\quad 0.044 \quad\quad\quad 0.049 \quad\quad\quad\quad 0.052\\
				      &  $\rho_3$ &\quad 0.643 \quad\quad\quad 0.343 \quad\quad\quad\quad0.302\\
				      &  $\rho_4$ &\quad 0.836 \quad\quad\quad 0.579 \quad\quad\quad\quad0.518\\
				      &  $\rho_5$ &\quad 0.924 \quad\quad\quad 0.792 \quad\quad\quad\quad0.722\\
\hline
\end{tabular}
\end{center}
\caption{ Empirical Rejection prob. for Nonparametric Specification for $\sol_2$ with $c_U=0.8$}
\end{table}

\section{Conclusion}
Based on the methodology of series estimation, we have developed in this paper a family of goodness-of-fit statistics and derived their asymptotic properties. The implementation of these statistics is straightforward. We have seen that the asymptotic results depend crucially on the choice of the smoothing operator $L$.
By choosing a stronger decaying sequence $\tau$, our test becomes more powerful with respect to local alternatives but might lose desirable consistency properties. We gave heuristic arguments how to choose the weights in practice. In addition, in a Monte Carlo investigation our tests perform well in finite samples.

\appendix
\section{Appendix}\label{app:proofs}
Throughout the Appendix, let $C>0$ denote a generic constant that may be different in different uses. For ease of notation let $\sum_i=\sum_{i=1}^n$ and $\sum_{i'<i}=\sum_{i=1}^n\sum_{i'=1}^{i-1}$.
Given $m\geq 1$, $\cE_m$ and $\cF_m$ denote the subspace of $\cL^2_Z$ and $\cL^2_W$ spanned by the functions $\{e_j\}_{j=1}^m$ and
$\{f_l\}_{l=1}^m$, respectively. $E_m$ and $E_m^\perp$ (resp. $F_m$ and $F_m^\perp$) denote the orthogonal projections on $\cE_m$ (resp. $\cF_m$) and its orthogonal complement
$\cE_m^\perp$ (resp. $\cF_m^\perp$), respectively. Respectively, given functions $\phi\in \cL^2_Z$ and  $\psi\in \cL^2_W$ we define by $[\phi]_{\um}$ and $[\psi]_{\um}$ $m$-dimensional vectors with entries 
$[\phi]_{j}=\Ex[\phi(Z)e_j(Z)]$ and $[\psi]_{l}=\Ex[\psi(W)f_l(W)]$ for $1\leq j,l\leq m$.
\subsection{Proofs of Section 2.}
\begin{proof}[\textsc{Proof of Theorem \ref{theo:norm:ind:main}.}]
Under $H_0$ we have $(Y_i-\sol_0(Z_i))f_{\um}^\tau(W_i)=U_if_{\um}^\tau(W_i)$ for all $m\geq 1$ and consequently we observe 
\begin{equation*}
\varsigma_\m^{-1}\big(nS_n-\mu_\m\big)
=\frac{1}{\varsigma_\m n}\sum_i\sum_{j=1}^\m\big(|U_if_j^\tau(W_i)|^2-{\textsl s}_{jj}\big)
+\frac{1}{\varsigma_\m n}\sum_{i\neq i'}\sum_{j=1}^\m U_iU_{i'}f_j^\tau(W_i)f_j^\tau(W_{i'})
\end{equation*}
where the first summand tends in probability to zero as $n\to\infty$. Indeed, since $\Ex |Uf_j(W)|^2-\varsigma_{jj} =0$, $j\geq 1$,  it holds for all $m\geq 1$
\begin{equation*}
 \frac{1}{(\varsigma_m n)^{2}}\Ex\big|\sum_i\sum_{j=1}^m|U_if_j^\tau(W_i)|^2-{\textsl s}_{jj}\big|^2
=\frac{1}{n\varsigma_m^{2}}\Ex\big|\sum_{j=1}^m|Uf_j^\tau(W)|^2-{\textsl s}_{jj}\big|^2\leq \frac{1}{n\varsigma_m^{2}}\Ex\|Uf_\um^\tau(W)\|^4.
\end{equation*}
By using Assumptions \ref{Ass:A1} and \ref{Ass:A2}, i.e., $\sup_{j\in \N}\Ex|f_j(W)|^4\leqslant \eta_f\,\eta_p$ and  $\Ex[U^4|W]\leq \sigma^4$, we conclude
\begin{equation}\label{norm:ind:ineq1}
 \Ex\|Uf_\um^\tau(W)\|^4\leq \max_{1\leq j\leq m}\Ex|Uf_j(W)|^4\Big(\sum_{j=1}^m\tau_j\Big)^2\leq \eta_f\,\eta_p\,\sigma^4 \Big(\sum_{j=1}^m\tau_j\Big)^2.
\end{equation}
Let $m=\m$ satisfy condition \eqref{cond:theo:norm} then $\Ex\|Uf_{\umn}^\tau(W)\|^4=o\big(n\varsigma_\m^{2}\big)$.  
Therefore, it is sufficient to prove
\begin{equation}\label{norm:result:ind}
\sqrt 2(\varsigma_\m n)^{-1}\sum_{i\neq i'}\sum_{j=1}^\m U_iU_{i'}f_j^\tau(W_i)f_j^\tau(W_{i'})\stackrel{d}{\rightarrow} \cN(0,1).
\end{equation}
Since $\varsigma_\m=o(1)$ this follows from Lemma A.2 and thus, completes the proof.
\end{proof}

\begin{proof}[\textsc{Proof of Theorem \ref{theo:norm:ind:main:tau}.}]
Similarly to the proof of Theorem \ref{theo:norm:ind:main} it is sufficient to study the asymptotic behavior of $n^{-1}\sum_{j=1}^\m\sum_{i\neq i'} U_iU_{i'}f_j^\tau(W_i)f_j^\tau(W_{i'})$.
For any finite  $m\geq 1$ we obtain
\begin{multline*}
\Ex\Big| \frac{1}{n}\sum_{j=1}^m\sum_{i\neq i'} U_iU_{i'}f_j^\tau(W_i)f_j^\tau(W_{i'})
-\frac{1}{n}\sum_{j=1}^\infty\sum_{i\neq i'} U_iU_{i'}f_j^\tau(W_i)f_j^\tau(W_{i'})\Big|^2\\
\leq \Ex\Big[\Ex[U_1^2U_{2}^2|W_1,W_2]\,\Big(\sum_{j>m}f_j^\tau(W_1)f_j^\tau(W_2)\Big)^2\Big]
\leq \sigma^4\eta_p\Big(\sum_{j> m}\tau_j\Big)^2
\end{multline*}
which, since $\sum_{j\geq 1}\tau_j=O(1)$, becomes sufficiently small (depending on $m$).
Note that $\big(\frac{1}{\sqrt n}\sum_{i} U_if_1^\tau(W_i),\dots, \frac{1}{\sqrt n}\sum_{i} U_if_m^\tau(W_i)\big)\stackrel{d}{\rightarrow} \mathcal N(0,\varSigma_m)$. Hence, for any finite $m\geq 1$ we have 
\begin{equation*}
\sum_{j=1}^m\Big|\frac{1}{\sqrt n}\sum_{i} U_if_j^\tau(W_i)\Big|^2\stackrel{d}{\rightarrow}\sum_{j=1}^m\lambda_{j}\chi_{1j}^2
\end{equation*}
with $\lambda_{j}$, $1\leq j\leq m$, being eigenvalues of $\varSigma_m$.
Moreover, we conclude for $m\geq 1$
\begin{multline*}
  \frac{1}{n}\sum_{j=1}^m\sum_{i\neq i'} U_iU_{i'}f_j^\tau(W_i)f_j^\tau(W_{i'})^t=\sum_{j=1}^m \Big(\Big|\frac{1}{\sqrt n}\sum_{i} U_if_j^\tau(W_i)\Big|^2-\frac{1}{n}\sum_{i} |U_if_j^\tau(W_i)|^2\Big)\\
\stackrel{d}{\rightarrow} \sum_{j=1}^m\big(\lambda_{j}\chi_{1j}^2-{\textsl s}_{jj}\big).
\end{multline*}
 It is easily seen that $\sum_{j=1}^m(\lambda_{j}\chi_{1j}^2-{\textsl s}_{jj})$ has expectation zero.
Hence, following the lines of page 198-199 of \cite{Serfling} we obtain that $\sum_{j>m}\big(\lambda_{j}\chi_{1j}^2-{\textsl s}_{jj}\big)$ becomes sufficiently small (depending on $m$) and thus, completes the proof.
\end{proof}

\begin{proof}[\textsc{Proof of Proposition  \ref{coro:norm:ind:main}.}]
For ease of notation let $\delta_n(\cdot):=\varsigma_\m^{1/2}n^{-1/2}\delta(\cdot)$.
Under the sequence of alternatives \eqref{loc:alt:ind} the following decomposition holds true
\begin{multline*}
S_n
=\big\|n^{-1}\sum_iU_if_\umn^\tau(W_i)\big\|^2
+2\skalarV{n^{-1}\sum_iU_if_\umn^\tau(W_i),n^{-1}\sum_i\delta_n(Z_i)f_\umn^\tau(W_i)}\\\hfill
+\big\|n^{-1}\sum_i\delta_n(Z_i)f_\umn^\tau(W_i)\big\|^2
=:I_n+2II_n+III_n.
\end{multline*}
Due to Theorem \ref{theo:norm:ind:main} we have $(\sqrt 2\varsigma_\m)^{-1}\big(n\,I_n-\mu_\m\big)\stackrel{d}{\rightarrow} \cN(0,1)$. Consider $II_n$. We observe
\begin{multline*}
 n\Ex|II_n|\leq \sum_{j=1}^\m\tau_j\big(\Ex|Uf_j(W)|^2\Ex|\delta_n(Z)f_j(W)|^2\big)^{1/2}
+\Big(n\Ex\Big|\sum_{j=1}^\m \tau_j[T\delta_n]_jUf_j(W)\Big|^2\Big)^{1/2}\\
\leq \sigma\sum_{j=1}^\m\tau_j\big(\Ex|\delta_n(Z)f_j(W)|^2\big)^{1/2}
+\sigma\eta_p\sqrt n\|\op\delta_n\|_W.
\end{multline*}
From the definition of $\delta_n$ and condition \eqref{cond:theo:norm} we infer that $n\Ex|II_n|=o(\varsigma_\m)$. Consider $III_n$. Employing again the definition of $\delta_n$ it is easily seen that $n\varsigma_\m^{-1}III_n=\sum_{j=1}^\m\tau_j[\op\delta]_j^2+o_p(1)$.
We conclude $(\sqrt 2\varsigma_\m)^{-1}nIII_n=2^{-1/2}\sum_{j\geq 1} \delta_j^2+o_p(1)$, which completes the proof.
\end{proof}

\begin{proof}[\textsc{Proof of Proposition  \ref{coro:norm:ind:main:tau}.}]
Let $\delta_n(\cdot):=n^{-1/2}\delta(\cdot)$. Similarly to the proof of Theorem \ref{theo:norm:ind:main:tau} it is straightforward to see that 
under the sequence of alternatives \eqref{loc:alt:ind:2} it holds
\begin{multline*}
 \frac{1}{n}\sum_{i\neq i'}\sum_{j=1}^\m (U_i+\delta_n(Z_i))(U_{i'}+\delta_n(Z_{i'}))f_j^\tau(W_i)f_j^\tau(W_{i'})\\
=\sum_{j=1}^\infty\Big(\Big|\frac{1}{\sqrt n}\sum_{i} U_if_j^\tau(W_i)+\frac{1}{ n}\sum_{i} \delta(Z_i)f_j^\tau(W_i)\Big|^2-\frac{1}{n}\sum_{i} 
\big|U_if_j^\tau(W_i)\big|^2\Big)\\
\stackrel{d}{\rightarrow} \sum_{j=1}^\infty\lambda_j\chi_{1j}({\delta_j}/{\lambda_j})
\end{multline*}
similar to the lines of page 198-199 of \cite{Serfling} and hence the assertion follows.
\end{proof}

\begin{proof}[\textsc{Proof of Proposition \ref{prop:cons}.}]
If $H_0$ fails we observe that $\sum_{j=1}^\infty\tau_j[\op(\sol-\sol_0)]_j^2=\int_{\cW}|LT(\sol-\sol_0)(w)p_W(w)/\nu(w)|^2\nu(w)dw\geq C\|LT(\sol-\sol_0)\|_W^2>0$ since $p_W/\nu$ is uniformly bounded from zero and $LT$ is nonsingular.
Now since $\varsigma_\m\alpha_n+\mu_\m=o(n)$ it is sufficient to show $S_n=\sum_{j=1}^\infty\tau_j[\op(\sol-\sol_0)]_j^2+o_p(1)$. We make use of the decomposition
\begin{multline*}
 S_n=\sum_{j=1}^\m\tau_j\big|n^{-1}\sum_i(Y_i-\sol_0(Z_i))f_j(W_i)-[\op(\sol-\sol_0)]_j\big|^2\\\hfill
+2\sum_{j=1}^\m\tau_j\big(n^{-1}\sum_i(Y_i-\sol_0(Z_i))f_j(W_i)-[\op(\sol-\sol_0)]_j\big)[\op(\sol-\sol_0)]_j+\sum_{j=1}^\m\tau_j[\op(\sol-\sol_0)]_j^2\\
=I_n+II_n+III_n.
\end{multline*}
Due to condition $\Ex|Y-\sol_0(Z)|^4<\infty$ it is easily seen that $I_n+II_n=o_p(1)$, which proves the result.
\end{proof}

\begin{proof}[\textsc{Proof of Proposition \ref{prop:unf:cons}.}]
 We make use of the decomposition
\begin{multline*}
 \PP\Big((\sqrt 2\,\varsigma_\m)^{-1}\big(n\,S_n-\mu_\m\big)>q_{1-\alpha}\Big)\\\hfill
\geq \PP\Big(\big\|n^{-1/2}\sum_i(\sol(Z_i)-\sol_0(Z_i))f_\umn^\tau(W_i)\big\|^2+\big\|n^{-1/2}\sum_iU_if_\umn^\tau(W_i)\big\|^2-\mu_\m\\
>\sqrt 2\,\varsigma_\m q_{1-\alpha}+2|\skalarV{n^{-1}\sum_i(\sol(Z_i)-\sol_0(Z_i))f_\umn^\tau(W_i),\sum_iU_if_\umn^\tau(W_i)}|\Big).
\end{multline*}
Uniformly over all $\sol\in\cG_n^\sr$ it holds
\begin{equation*}
 \skalarV{n^{-1}\sum_i(\sol(Z_i)-\sol_0(Z_i))f_\umn^\tau(W_i),\sum_iU_if_\umn^\tau(W_i)}=O_p\big(\max(\sqrt n\|L\op(\sol-\sol_0)\|_W,\varsigma_\m)\big).\label{proof:prop:unf:cons:1}
\end{equation*}
Indeed, this is easily seen from
\begin{equation*}
 \Ex\big|\sum_{j=1}^\m\tau_j\Ex[(\sol(Z)-\sol_0(Z))f_j(W)]\sum_iU_if_j(W_i)\big|^2
\leq \sigma^2\eta_p n\sum_{j=1}^\m\Ex[(\sol(Z)-\sol_0(Z))f_j^\tau(W)]^2
\end{equation*}
and further, denoting $\psi_{ji}=(\sol(Z_i)-\sol_0(Z_i))f_j(W)-\Ex[(\sol(Z)-\sol_0(Z))f_j(W)]$, $1\leq j\leq \m$, $1\leq i\leq n$, from
\begin{multline*}
 \Ex\big|n^{-1}\sum_{i\neq i'}\sum_{j=1}^\m\tau_j\psi_{ji}U_{i'}f_j(W_{i'})\big|^2
=\frac{n-1}{n}\sum_{j,j'=1}^\m\tau_j\tau_{j'}\Ex\big[\psi_{j1}\psi_{j'1}\big]\Ex\big[U^2f_j(W)f_{j'}(W)\big]\\
\leq C\sum_{j,j'=1}^\m\tau_j\tau_{j'}\Ex\big[U^2f_j(W)f_{j'}(W)\big]\leq C \sigma^2 \Ex\big|\sum_{j=1}^\m\tau_jf_j(W)\big|^2=O\Big(\sum_{j=1}^\m\tau_j^2\Big).
\end{multline*}
Thereby, for all $0<\varepsilon'<1$ there exists some constant $C>0$ such that
\begin{multline*}
 \PP\Big((\sqrt 2\,\varsigma_\m)^{-1}\big(n\,S_n-\mu_\m\big)>q_{1-\alpha}\Big)\\
\geq \PP\Big(\big\|n^{-1/2}\sum_i(\sol(Z_i)-\sol_0(Z_i))f_\umn^\tau(W_i)\big\|^2+\big\|n^{-1/2}\sum_iU_if_\umn^\tau(W_i)\big\|^2-\mu_\m\\
>\sqrt 2\,\varsigma_\m q_{1-\alpha}+C\max(\sqrt n\|L\op(\sol-\sol_0)\|_W,\varsigma_\m)\Big)-\varepsilon'.
\end{multline*}
Note that  $\big\|n^{-1/2}\sum_iU_if_\umn^\tau(W_i)\big\|^2=\mu_\m+O_p(\varsigma_\m)$ due to Theorem \ref{theo:norm:ind:main}. Moreover,
\begin{multline*}
\big\|n^{-1/2}\sum_i(\sol(Z_i)-\sol_0(Z_i))f_\umn^\tau(W_i)\big\|^2
\geq n\sum_{j=1}^\m\tau_j[\op(\sol-\sol_0)]_j^2\\
-2\big|\skalarV{\sum_i(\sol(Z_i)-\sol_0(Z_i))f_\umn^\tau(W_i)-n[L\op(\sol-\sol_0)]_\umn, [L\op(\sol-\sol_0)]_\umn}=I_n+II_n.
\end{multline*}
Consider $II_n$. For $1\leq j\leq \m$ let $s_j=\tau_j[\op(\sol-\sol_0)]_j/\big(\sum_{j=1}^\infty\tau_j[\op(\sol-\sol_0)]_j^2\big)^{1/2}$ then clearly $\sum_{j=1}^\m s_j^2=1$ and thus $\Ex|\sum_{j=1}^\m s_jf_j(W)|^2\leq\eta_f\eta_p$. Further, since $\sup_{z\in\cZ}|\sol(z)-\sol_0(z)|^2\leq C$ we calculate
\begin{multline*}
\Ex|II_n|^2=n\Ex\Big|\sum_{j=1}^\m\tau_j \big((\sol(Z)-\sol_0(Z))f_j(W)-[\op(\sol-\sol_0)]_j\big) [\op(\sol-\sol_0)]_j\Big|^2\\
\leq n\sum_{j=1}^\m\tau_j[\op(\sol-\sol_0)]_j^2 \Ex\Big|\sum_{j=1}^\m s_j(\sol(Z)-\sol_0(Z))f_j(W)\Big|^2
=O\big(n\|L\op(\sol-\sol_0)\|_W^2\big) 
\end{multline*}
and hence $II_n=O_p(\sqrt n\|L\op(\sol-\sol_0)\|_W)$. 
Consider $I_n$. Note that $\|LT(\sol-\sol_0)\|_W^2\leq C$ for all $\sol\in\cG_n^\sr$ we have $I_n\geq Cn\|LT(\sol-\sol_0)\|_W^2$ for $n$ sufficiently large. 
Since on $\cG_n^\sr$ we have $n\|LT(\sol-\sol_0)\|_W^2\geq \sr\,\varsigma_\m$ we obtain the result by choosing $\sr$ sufficiently large.
\end{proof}
\subsection{Proofs of Section 3.}
For ease of notation, we write in the following $\phi(\cdot)$ for $\phi(\cdot,\thet)$ and  $\phi_{\vartheta_l}(\cdot)$ for $\phi_{\vartheta_l}(\cdot,\thet)$.
\begin{proof}[\textsc{Proof of Theorem \ref{prop:par}.}]
The proof is based on the decomposition under $H_{\text{p}}$
 \begin{multline}\label{proof:ex:dec}
  S_n^{\text{p}}=\big\|n^{-1}\sum_iU_if_\umn^\tau(W_i)\big\|^2
+2\skalarV{n^{-1}\sum_iU_if_\umn^\tau(W_i),n^{-1}\sum_i\big(\phi(Z_i)-\phi(Z_i,\hthet)\big)f_\umn^\tau(W_i)}\\
+\|n^{-1}\sum_i\big(\phi(Z_i)-\phi(Z_i,\hthet)\big)f_\umn^\tau(W_i)\|^2
=I_n+2II_n+III_n.
 \end{multline}
Due to Theorem \ref{theo:norm:ind:main} it holds $(\sqrt2\varsigma_\m)^{-1}(nI_n-\mu_\m)\stackrel{d}{\rightarrow}\cN(0,1)$. Consider $III_n$. It holds $\phi(Z_i)-\phi(Z_i,\hthet)=\phi_\vartheta(Z_i)^t (\thet-\hthet)+(\thet-\hthet)^t\phi_{\vartheta\vartheta}(Z_i, \overline\vartheta_n)^t (\thet-\hthet)/2$ for some $\overline\vartheta_n$ between $\hthet$ and $\thet$. From the bounds imposed in  Assumption \ref{Ass:par} $(ii)$ we infer
\begin{equation*}
nIII_n
\leq 2n\|\thet-\hthet\|^2\Big(\sum_{l=1}^k\sum_{j=1}^\m\tau_j [\op\phi_{\vartheta_l}]_{j}^2
+\sum_{l=1}^k\sum_{j=1}^\m\tau_j\big(\frac{1}{n}\sum_i\phi_{\vartheta_l}(Z_i) f_j(W_i)-[\op\phi_{\vartheta_l}]_{j}\big)^2\Big)+o_p(1).
\end{equation*}
For each $1\leq l\leq k$  we have  
\begin{multline}\label{proof:main:ineq}
 \sum_{j=1}^\m[\op\phi_{\vartheta_l}]_{j}^2=\sum_{j=1}^\m\Big(\int_{\cW}(T\phi_{\vartheta_l})(w)f_j(w)p_W(w)dw\Big)^2
\leq \int_{\cW}|(T\phi_{\vartheta_l})(w)p_W(w)/\nu(w)|^2\nu(w)dw\\
\leq\eta_p\|\op \phi_{\vartheta_l}\|_W^2\leq\eta_p\Ex|\phi_{\vartheta_l}(Z,\thet)|^2\leq \eta_p\eta_\phi
\end{multline}
 by applying Jensen's inequality.
Moreover, we calculate
\begin{equation}\label{proof:prop:par:1}
 \sum_{l=1}^k\sum_{j=1}^\m\Ex\big|\frac{1}{n}\sum_i\phi_{\vartheta_l}(Z_i) f_j(W_i)- [\op\phi_{\vartheta_l}]_{j}\big|^2\leq \frac{k\m}{n}\sup_{j,l\geq 1}\Ex|\phi_{\vartheta_l}(Z) f_j(W)|^2\leq \eta^4\frac{k\m}{n}.
\end{equation}
These estimates together with $\|\thet-\hthet\|=O_p(n^{-1/2})$ imply $nIII_n=o_p(\varsigma_\m)$.
We are left with the proof of $nII_n=o_p(\varsigma_\m)$.
We observe for each $1\leq l\leq k$
\begin{multline*}
\Ex\Big|\sum_{j=1}^\m\tau_j \Big(n^{-1/2}\sum_iU_if_j(W_i)\big(n^{-1}\sum_i\phi_{\vartheta_l}(Z_i) f_j(W_i)- [\op\phi_{\vartheta_l}]_{j}\big)\Big)\Big|\\
\leq  n^{-1/2}\sum_{j=1}^\m \tau_j\big(\Ex|Uf_j(W)|^2\big)^{1/2}\big(\Ex|\phi_{\vartheta_l}(Z) f_j(W)|^2\big)^{1/2}
= O\Big(n^{-1/2}\sum_{j=1}^\m \tau_j\Big)=o(\varsigma_\m).
\end{multline*}
 Now since $n^{1/2 }(\thet-\hthet)=O_p(1)$ we infer
\begin{equation*}
nII_n=n^{1/2 }(\thet-\hthet)^t\sum_{j=1}^\m\tau_j \Big(\varsigma_\m^{-1}n^{-1/2}\sum_iU_if_j(W_i)\Ex[\phi_{\vartheta}(Z) f_j(W)]\Big)+o_p(1).
\end{equation*}
We observe for each $1\leq l\leq k$
\begin{equation*}
 \varsigma_\m^{-2}n^{-1}\Ex\Big|\sum_{j=1}^\m\tau_j\sum_iU_if_j(W_i)[\op\phi_{\vartheta_l}]_{j}\Big|^2
\leq \varsigma_\m^{-2}\sigma^2\eta_p\sum_{j=1}^\m [\op\phi_{\vartheta_l}]_{j}^2
\leq \varsigma_\m^{-2}\sigma^2\,\eta_p^2\,\eta_f
\end{equation*}
which implies $nII_n=o_p(\varsigma_\m)$ and thus, in light of decomposition \eqref{proof:ex:dec}, completes the proof.
\end{proof}

\begin{proof}[\textsc{Proof of Theorem \ref{prop:par:tau}.}]
For $1\leq j\leq \m$ we make use of the following decomposition
\begin{multline}\label{proof:prop:par:tau:1}
 n^{-1/2}\sum_if_j(W_i)\Big(U_i+\phi(Z_i)-\phi(Z_i,\hthet)\Big)
=n^{-1/2}\sum_i\Big(f_j(W_i)U_i-\sum_{l=1}^k [\op\phi_{\vartheta_l}]_{j}h_l(V_i)\Big)\\\hfill
+\sum_{l=1}^k\Big(n^{-1}\sum_if_j(W_i) \phi_{\vartheta_l}(Z_i)-[\op\phi_{\vartheta_l}]_{j}\Big) \Big(n^{-1/2}\sum_ih_l(V_i)\Big)\\
+\sum_{l=1}^kn^{-1}\sum_if_j(W_i) \phi_{\vartheta_l}(Z_i) r_l+o_p(1)
=A_{nj}+B_{nj}+C_{nj}+o_p(1)
\end{multline}
where $r_\uk=(r_1,\dots,r_k)^t$ is a stochastic vector satisfying $r_\uk=o_p(1)$. Consequently, under $H_{\text{p}}$ we have
\begin{equation*}
 nS_n^{\text{p}}=\sum_{j=1}^\m \tau_jA_{nj}^2+2\sum_{j=1}^\m \tau_jA_{nj}(B_{nj}+C_{nj})+\sum_{j=1}^\m\tau_j(B_{nj}+C_{nj})^2+o_p(1).
\end{equation*}
Clearly, for all $1\leq i\leq n$ the random variables $U_if_j^\tau(W_i)+\Ex\big[f_j^\tau(W)\phi_{\vartheta}(Z)^t\big]h_\uk(V_i)$, $1\leq j\leq \m$, are centered with bounded second moment. Due to the proof of Theorem \ref{theo:norm:ind:main:tau} it is easily seen that $\sum_{j=1}^\m\tau_j A_{nj}^2 \stackrel{d}{\rightarrow}\sum_{j=1}^\infty\lambda_j^\textsl{p}\,\chi_{1j}^2$. Inequality \eqref{proof:prop:par:1} yields $\sum_{j=1}^\m B_{nj}^2=o_p(1)$. Since $\sum_{j=1}^\m[\op\phi_{\vartheta}]_{j}^2\leq\eta_p\eta_\phi$ we have $\|\Ex[f_\umn(W)\phi_{\vartheta}(Z)^t]r_\uk\|^2\leq k\,\eta_p \eta_\phi\|r_\uk\|^2=o_p(1)$ and hence $\sum_{j=1}^\m C_{nj}^2=o_p(1)$.
Finally, the Cauchy-Schwarz inequality implies 
$\sum_{j=1}^\m\tau_j A_{nj}(B_{nj}+C_{nj})=o_p(1)$,
 which completes the proof.
\end{proof}

\begin{proof}[\textsc{Proof of Proposition  \ref{coro:norm:par}.}]
Without loss of generality we may assume $\delta=\delta_\perp$ (otherwise replace $\phi(Z_i)$ by $\phi(Z_i)+\phi_\vartheta(Z_i)^t\Ex[\delta(Z)\phi_\vartheta(Z)]$.
Consider the case $\varsigma_\m^{-1}=o(1)$. Under the sequence of alternatives \eqref{loc:alt:ind} the following decomposition holds true
\begin{multline*}
S_n^\text p
=\big\|n^{-1}\sum_i(U_i+\varsigma_\m^{1/2}n^{-1/2}\delta_\perp(Z_i))f_\umn^\tau(W_i)\big\|^2\\\hfill
+2\skalarV{n^{-1}\sum_i(U_i+\varsigma_\m^{1/2}n^{-1/2}\delta_\perp(Z_i))f_\umn^\tau(W_i),n^{-1}\sum_i(\phi(Z_i)-\phi(Z_i,\hthet))f_\umn^\tau(W_i)}\\\hfill
+\big\|n^{-1}\sum_i(\phi(Z_i)-\phi(Z_i,\hthet))f_\umn^\tau(W_i)\big\|^2.
\end{multline*}
Due to Proposition \ref{coro:norm:ind:main} and the proof of Theorem \ref{prop:par} it is sufficient to show 
\begin{equation}\label{eq:coro:norm:par}
 \skalarV{n^{-1}\sum_i\delta_\perp(Z_i)f_\umn^\tau(W_i),n^{-1/2}\sum_i(\phi(Z_i)-\phi(Z_i,\hthet))f_\umn^\tau(W_i)}=o_p(\sqrt{\varsigma_\m}).
\end{equation}
Indeed, since $\delta_{j \perp}=\sqrt{\tau_j}\,\Ex[\delta_\perp(Z)f_j(W)]$ we have
\begin{multline*}
 \sum_{j=1}^\m\delta_{j \perp}n^{-1/2}\sum_i(\phi(Z_i)-\phi(Z_i,\hthet))f_j(W_i)
=\sqrt n(\thet-\hthet)^t\sum_{j=1}^\m\delta_{j \perp}\Ex[\phi_{\vartheta}(Z)f_j(W)]+o_p(1)\\
\leq \eta_p\,\eta_\phi\,\sqrt n\|\thet-\hthet\|\,\sum_{j=1}^\infty\delta_{j \perp}^2+o_p(1)=O_p(1)
\end{multline*}
and hence \eqref{eq:coro:norm:par} holds true.

Consider the case $\sum_{j=1}^\m\tau_j=O(1)$. We make use of decomposition \eqref{proof:prop:par:tau:1} where $U_i$ is replaced by $U_i+n^{-1/2}\delta_\perp(Z_i)$. Similarly to the proof of Proposition  \ref{coro:norm:ind:main:tau} it is easily seen that $\sum_{j=1}^\m\tau_j A_{nj}^2\stackrel{d}{\rightarrow}\sum_{j=1}^\infty\lambda_j^{\textsl{p}}\,\chi_{1j}^2(\delta_{j \perp}/\lambda_j^{\text{p}})$. Thereby, due to the proof of Theorem \ref{prop:par:tau}, the assertion follows.
\end{proof}

\begin{proof}[\textsc{Proof of Proposition \ref{prop:par:cons}.}] 
Consider first the case $\varsigma_\m^{-1}=o(1)$.
Similar to the proof of Theorem \ref{prop:par} we observe that $\|n^{-1}\sum_i(\phi(Z_i,\vartheta_0)-\phi(Z_i,\hthet))f_\umn^\tau(W_i)\|^2=o_p(1)$  and  $\|n^{-1}\sum_i(Y_i-\phi(Z_i,\vartheta_0))f_\umn^\tau(W_i)\|^2=\sum_{j=1}^\infty\tau_j[\op(\sol-\phi(\cdot,\vartheta_0))]_j^2+o_p(1)$. Thus, the result follows as in the proof of Proposition \ref{prop:cons}.
In case of $\sum_{j=1}^\m\tau_j=O(1)$, we obtain similarly that $S_n^\text p=\sum_{j=1}^\m\tau_j\big|n^{-1}\sum_i\big((Y_i-\phi(Z_i,\vartheta_0))f_j(W_i)\big|^2+o_p(1)$ and hence, $S_n^\text p=\sum_{j=1}^\infty\tau_j[\op(\sol-\phi(\cdot,\vartheta_0))]_j^2+o_p(1).$
\end{proof}

\begin{proof}[\textsc{Proof of Proposition \ref{prop:par:unf:cons}.}]
Consider the case $\varsigma_\m^{-1}=o(1)$. 
The basic inequality $(a-b)^2\geq a^2/2-b^2$, $a,b\in\mathbb R$,  yields
\begin{multline}\label{prop:par:unf:cons:eq}
 \PP\Big((\sqrt 2\,\varsigma_\m)^{-1}\big(n\,S_n^{\text{p}}-\mu_\m\big)>q_{1-\alpha}\Big)\\\hfill
\geq \PP\Big(1/2\big\|n^{-1/2}\sum_i(\sol(Z_i)-\phi(Z_i,\vartheta_0))f_\umn^\tau(W_i)\big\|^2+\big\|n^{-1/2}\sum_iU_if_\umn^\tau(W_i)\big\|^2-\mu_\m\\\hfill
>\sqrt 2\,\varsigma_\m q_{1-\alpha}+2|\skalarV{n^{-1}\sum_i(\sol(Z_i)-\phi(Z_i,\hthet))f_\umn^\tau(W_i),\sum_iU_if_\umn^\tau(W_i)}|\\
+\big\|n^{-1/2}\sum_i(\phi(Z_i,\vartheta_0)-\phi(Z_i,\hthet))f_\umn^\tau(W_i)\big\|^2\Big).
\end{multline}
From the proof of Theorem \ref{prop:par} we infer $\big\|n^{-1/2}\sum_i(\phi(Z_i,\hthet)-\phi(Z_i,\vartheta_0))f_\umn^\tau(W_i)\big\|^2=o_p(\varsigma_\m)$ and
\begin{multline*}
 \skalarV{n^{-1}\sum_i(\sol(Z_i)-\phi(Z_i,\hthet))f_\umn^\tau(W_i),\sum_iU_if_\umn^\tau(W_i)}\\
=\skalarV{n^{-1}\sum_i(\sol(Z_i)-\phi(Z_i,\vartheta_0))f_\umn^\tau(W_i),\sum_iU_if_\umn^\tau(W_i)}+o_p(\varsigma_\m).
\end{multline*}
Thus, following line by line the proof of Proposition \ref{prop:unf:cons}, the assertion follows. In case of $\sum_{j=1}^\m\tau_j=O(1)$ the assertion follows similarly.
\end{proof}

\subsection{Proofs of Section 4.}
In the following, we denote $[\widehat Q]_\ukn=n^{-1}\sum_i e_\ukn(Z_i)e_\ukn(Z_i)^t$.
By Assumption \ref{Ass:ex}, the eigenvalues of $\Ex[ e_\ukn(Z)e_\ukn(Z)^t]$ are bounded away from zero and hence, it may be assumed  that $\Ex[ e_\ukn(Z)e_\ukn(Z)^t]=I_\k$ (cf. \cite{Newey1997}, p. 161).
\begin{proof}[\textsc{Proof of Theorem \ref{prop:ex}.}]
The proof is based on the decomposition \eqref{proof:ex:dec} where the estimator $\phi(\cdot,\hthet)$ is replaced by $\overline\sol_\k(\cdot)$ given in \eqref{est:ex}.
It holds $nIII_n=o_p(\varsigma_\m)$, which can be seen as follows. We make use of
\begin{multline*}
 III_n/2\leq \big\|\frac{1}{n}\sum_i(E_\k\sol_0(Z_i)-\osol_\k(Z_i))f_\umn^\tau(W_i)\big\|^2
+\big\|\frac{1}{n}\sum_i\big(E_\k^\perp\sol_0\big)(Z_i)f_\umn^\tau(W_i)\big\|^2\\
=: A_{n1}+A_{n2}.
\end{multline*}
Consider $A_{n1}$. We observe
\begin{multline}\label{proof:ex:an:dec}
 A_{n1}
\leq 2\|\Ex[f_\umn^\tau(W)e_\ukn(Z)^t][\widehat Q]_\ukn^{-1}([\widehat Q]_\ukn[\sol_0]_\ukn-n^{-1}\sum_iY_ie_\ukn(Z_i))\|^2\\\hfill
+2\|E_\k\sol_0-\osol_\k\|_Z^2\sum_{j=1}^\m\tau_j\sum_{l=1}^\k|n^{-1}\sum_ie_{l}(Z_i)f_j(W_i)-[\op]_{jl}|^2\\
=:2B_{n1}+2B_{n2}.
\end{multline}
 For $B_{n1}$ we evaluate due to the relation $[\widehat Q]_\ukn^{-1}=I_\k-[\widehat Q]_\ukn^{-1}([\widehat Q]_\ukn-I_\k)$ that
\begin{multline*}
 B_{n1}
\leq 2\big\|\Ex[f_\umn^\tau(W)e_\ukn(Z)^t]n^{-1}\sum_i(E_\k\sol_0(Z_i)-Y_i)e_\ukn(Z_i)\big\|^2\\\hfill
+2\big\|\Ex[f_\umn^\tau(W)e_\ukn(Z)^t]\big\|^2\big\|[\widehat Q]_\ukn-I_\k\big\|^2\,\big\|[\widehat Q]_\ukn^{-1}\big\|^2\,\big\|n^{-1}\sum_i(E_\k\sol_0(Z_i)-Y_i)e_\ukn(Z_i)\big\|^2.
\end{multline*}
Since the spectral norm of a matrix is bounded by its Frobenius norm it holds
\begin{equation*}
 \Ex\big\|[\widehat Q]_\ukn-I_\k\big\|^2\leq n^{-1}\sum_{l,l'=1}^\k\Ex|e_l(Z)e_{l'}(Z)|^2\leq \eta_e\,n^{-1}k_n^2.
\end{equation*}
Further, from $\Ex[(E_\k\sol_0(Z)-Y)e_\ukn(Z)]=0$ we deduce
\begin{multline*}
 \Ex\big\|\Ex[f_\umn^\tau(W)e_\ukn(Z)^t]n^{-1}\sum_i(E_\k\sol_0(Z_i)-Y_i)e_\ukn(Z_i)\big\|^2\\\hfill
\leq n^{-1}\sum_{j=1}^\m\Ex\big|\sum_{j=1}^\k \Ex[f_j(W)e_l(Z)](E_\k\sol_0(Z)-Y)e_l(Z)|^2\\
\leq C\eta_p\,n^{-1}\sum_{j=1}^\m\sum_{j=1}^\k \Ex[f_j(W)e_l(Z)]^2
=O(n^{-1}\k)
\end{multline*}
where we used the definition of $\cF_\sw$ and that $\Ex[U^2|Z]$ is bounded.
Moreover, since the difference of eigenvalues of $[\widehat Q]_\ukn$ and $I_\k$ is bounded by $\|[\widehat Q]_\ukn-I_\k\|$, the smallest eigenvalue of $[\widehat Q]_\ukn$ converges in probability to one and hence, $\|[\widehat Q]_\ukn^{-1}\|^2=1+o_p(1)$. Further, note that $\|\Ex[f_\umn^\tau(W)e_\ukn(Z)^t]\|^2\leq \sum_{j=1}^\m\sum_{j=1}^\k \Ex[f_j(W)e_l(Z)]^2=O(\k)$. 
Consequently, 
\begin{equation}\label{proof:ex:eq:1}
n\|E_\k\sol_0-\osol_\k\|_Z^2=O_p(\k)
\end{equation}
and since $\k=o(\varsigma_\m)$ we proved $nB_{n1}=o_p(\varsigma_\m)$. 
In addition, applying inequality \eqref{proof:prop:par:1} together with equation \eqref{proof:ex:eq:1} yields $nB_{n2}=o_p(\varsigma_\m)$.
Consequently, $nA_{n1}=o(\varsigma_\m)$. Consider $A_{n2}$. Similar to the derivation of \eqref{proof:main:ineq} we obtain
\begin{equation*}
\Ex\big\|n^{-1}\sum_i\big(E_\k^\perp\sol_0\big)(Z_i)f_\umn^\tau(W_i)\big\|^2\leq 2\eta_p\|E_\k^\perp\sol_0\|_Z^2+2n^{-1}\sum_{j=1}^\m\Ex|E_\k^\perp\sol_0(Z)f_j(W)|^2.
\end{equation*}
We have
\begin{equation}\label{proof:ex:ineq:2}
 \sum_{j=1}^\m\tau_j\Ex|(E_\k^\perp\sol_0)(Z)f_j(W)|^2=O\Big(\sw_\k^{-1}\sum_{j=1}^\m\tau_j\Big)
=o(\varsigma_\m)
\end{equation}
and $n\|E_\k^\perp\sol_0\|_Z^2=O(n\sw_\k^{-1})=o(\varsigma_\m)$. Hence, $nIII_n=o_p(\varsigma_\m)$.
Consider $II_n$. We calculate
\begin{multline}\label{proof:ex:an:dec:bn}
 nII_n \leq
\Big|\sum_{j=1}^\m \tau_j\sum_i U_if_j(W_i)([\sol_0]_\ukn-[\osol]_\ukn)^t\Big(n^{-1}\sum_ie_\ukn(Z_i) f_j(W_i)-\Ex\big[ e_\ukn(Z) f_j(W)\big]\Big)\Big|\\
+\Big|\sum_{j=1}^\m\tau_j\sum_{l=1}^\k([\sol_0]_l-[\osol]_l) \Big(\sum_iU_if_j(W_i)[\op]_{jl}\Big)\Big|\\
+\Big|\sum_{j=1}^\m\tau_j \Big(\sum_i U_if_j(W_i)\Big)\Big( n^{-1}\sum_iE_\k^\perp\sol_0(Z_i) f_j(W_i)-\Ex [E_\k^\perp\sol_0(Z) f_j(W)] \Big)\Big|\\
+\Big|\sum_{j=1}^\m\tau_j \Big(\sum_i U_if_j(W_i)\Big)\Ex[E_\k^\perp\sol_0(Z) f_j(W)] \Big|=C_{n1}+C_{n2}+C_{n3}+C_{n4}.
\end{multline}
Consider $C_{n1}$. Applying the Cauchy-Schwarz inequality twice gives
\begin{equation*}
C_{n1}\leq \|E_\k\sol_0-\osol_\k\|_Z\sum_{j=1}^\m\tau_j\big|\sum_iU_if_j(W_i)\big|\,
\Big(\sum_{l=1}^\k|n^{-1}\sum_i e_l(Z_i) f_j(W_i)-\Ex[e_l(Z) f_j(W)]|^2\Big)^{1/2}.
\end{equation*}
From $\Ex|\sum_iU_if_j(W_i)|^2\leq  n\,\eta_f\sigma^2$, relation \eqref{proof:ex:eq:1}, and inequality \eqref{proof:prop:par:1} we infer $C_{n1}=o_p(\varsigma_\m)$ due to condition \eqref{cond:prop:ex}. For $C_{n2}$ we evaluate
\begin{equation*}
 C_{n2}\leq \|E_\k\sol_0-\osol_\k\|_Z\Big(\sum_{l=1}^\k \big|\sum_{j=1}^\m \sum_iU_if_j(W_i)[\op]_{jl}\big|^2\Big)^{1/2}.
\end{equation*} 
Now $\sum_{j=1}^\m \sum_{l=1}^\k[\op]_{jl}^2=O(\k)$ together with \eqref{proof:ex:eq:1} yields $C_{n2}=o_p(1)$. 
Consider $C_{n3}$. Since $\Ex[U^2|W]\leq \sigma^2$ we conclude similarly as in inequality \eqref{proof:ex:ineq:2} that
\begin{equation*}
 \Ex C_{n3}\leq \sum_{j=1}^\m\tau_j\big(\Ex |Uf_j(W)|^2\big)^{1/2} \big(\Ex |E_\k^\perp\sol_0(Z) f_j(W)|^2\big)^{1/2}\\
=O\Big(\sw_\k^{-1/2}\sum_{j=1}^\m\tau_j\Big)
=o(\varsigma_\m).
\end{equation*}
Consider $C_{n4}$. We calculate
\begin{equation*}
\Ex |C_{n4}|^2\leq n\,\eta_p\,\sigma^2\sum_{j=1}^\m [TE_\k^\perp\sol_0]_j^2\leq n\,\eta_p^2\,\sigma^2\|TE_\k^\perp\sol_0\|_W^2
=O(n\sw_\k^{-1})
=o(\varsigma_\m).
\end{equation*}
Consequently, in light of decomposition \eqref{proof:ex:an:dec:bn} we obtain $nII_n=o(\varsigma_\m)$, which completes the proof.
\end{proof}

\begin{proof}[\textsc{Proof of Theorem \ref{prop:ex:tau}.}]
Employing the equality $[\widehat Q]_\ukn^{-1}=I_\k-[\widehat Q]_\ukn^{-1}([\widehat Q]_\ukn-I_\k)$
we obtain for all $1\leq j\leq \m$
\begin{multline}\label{proof:prop:ex:tau:1}
 n^{-1/2}\sum_if_j(W_i)\Big(U_i+\sol_0(Z_i)-\osol_\k(Z_i)\Big)\\
=n^{-1/2}\sum_i\Big(f_j(W_i)U_i+\Ex[f_j(W)e_\ukn(Z)^t]e_\ukn(Z_i)\big(\sol_0(Z_i)-Y_i\big)\Big)\\\hfill
+n^{-1/2}\sum_i\Ex\big[f_j(W)e_\ukn(Z)^t\big][\widehat Q]_\ukn^{-1}\big([\widehat Q]_\ukn-I_\k\big)e_\ukn(Z_i)\big(E_\k\sol_0(Z_i)-Y_i\big)\\\hfill
+\Big(n^{-1}\sum_if_j(W_i) e_\ukn(Z_i)-\Ex\big[f_j(W)e_\ukn(Z)^t\big]\Big) \sqrt n \big([\sol_0]_\ukn-[\osol_\k]_\ukn\big)\\
-n^{-1/2}\sum_i \Ex[f_j(W)e_\ukn(Z)^t]e_\ukn(Z_i)E_\k^\perp\sol_0(Z_i)
=A_{nj}+B_{nj}+C_{nj}+D_{nj}.
\end{multline}
Due to Assumption \ref{Ass:ex} (ii) we may assume that $\{e_1,\dots,e_k\}$ forms an orthonormal system in $\cL_Z^2$ and hence $\sum_{l=1}^k\Ex[f_j(W)e_l(Z)]^2$ is bounded uniformly in $k$. Thereby, we conclude $\sum_{j=1}^\m \tau_j\sum_{l> \k}\Ex[f_j(W)e_l(Z)]e_l(\cdot)=o(1)$.
 Now following line by line the proof of Theorem \ref{theo:norm:ind:main:tau} we deduce 
\begin{equation*}
 \sum_{j=1}^\m \tau_jA_{nj}^2=\sum_{j=1}^m\tau_j\Ex\Big|n^{-1/2}\sum_iU_i\Big(f_j(W_i)+\sum_{l\geq 1}\Ex[f_j(W)e_l(Z)]e_l(Z_i)\Big)\Big|^2+o_p(1)
\stackrel{d}{\rightarrow} 
\sum_{j=1}^\infty\lambda_j^{\textsl{e}}\,\chi_{1j}^2.
\end{equation*}
Moreover, we see similarly to the proof of Theorem \ref{prop:ex} that $\sum_{j=1}^\m\tau_j( B_{nj}^2+C_{nj}^2+D_{nj}^2)=o_p(1)$, which completes the proof.
\end{proof}

\begin{proof}[\textsc{Proof of Lemma  \ref{lem:crit}.}]
 Note that the squared Frobenius norm of $\widehat\varSigma_\m-\varSigma_\m$ is given by 
\begin{multline*}
 \sum_{j,l=1}^\m\Big|n^{-1}\sum_{i}(Y_i-\osol_\k(Z_i))^2f_j^\tau(W_i)f_l^\tau(W_i)-{\textsl s}_{jl}\big|^2\\
\leq \|\osol_\k-E_\k\sol\|_Z^4\sum_{j,l=1}^\m\Ex\big[\|e_\ukn(Z)\|^2f_j^\tau(W)f_l^\tau(W)\big]^2\\\hfill
+\sum_{j,l=1}^\m\Ex\big[(E_\k^\perp\sol_0(Z))^2f_j^\tau(W)f_l^\tau(W)\big]^2+o_p(1)\\
\leq \|\osol_\k-E_\k\sol\|_Z^4O\Big(\big(\sum_{j=1}^\m\tau_j\big)^2\Big)+O\Big(\big(\sw_\k^{-1}\sum_{j=1}^\m\tau_j\big)^2\Big)+o_p(1)=o_p(1)
\end{multline*}
by using relation \eqref{proof:ex:eq:1}. Consequently, the Frobenius norm of $\widehat\varSigma_\m$ equals $\varsigma_\m+o_p(1)$. Consistency of the trace of $\widehat\varSigma_\m$ is seen similarly.
\end{proof}

\begin{proof}[\textsc{Proof of Proposition  \ref{coro:norm:ex}.}]
Similar to the proof of Proposition  \ref{coro:norm:par} it is sufficient to show
\begin{equation}\label{eq:coro:norm:ex}
 \skalarV{n^{-1}\sum_i\delta(Z_i)f_\umn^\tau(W_i),n^{-1/2}\sum_i(\sol_0(Z_i)-\osol_\k(Z_i))f_\umn^\tau(W_i)}=o_p(\sqrt{\varsigma_\m}).
\end{equation}
By employing Jensen's inequality and estimate \eqref{proof:ex:eq:1} we obtain
\begin{multline*}
 \sum_{j=1}^\m\tau_j[\op\delta]_j\frac{1}{\sqrt n}\sum_i(E_\k\sol_0(Z_i)-\osol_\k(Z_i))f_j(W_i)\\
\leq\sqrt n\|\op\delta\|_{\tau}\|T(E_\k\sol_0-\osol_\k)\|_W+o_p(1)=o_p(\varsigma_\m).
\end{multline*}
Similarly to the upper bounds of $C_{n3}$ and $C_{n4}$ in the proof of Theorem \ref{prop:ex} it is straightforward to see that $\sum_{j=1}^\m\tau_j[\op\delta]_jn^{-1/2}\sum_iE_\k^\perp\sol_0(Z_i)f_j(W_i)=o_p(\varsigma_\m)$ and, hence equation \eqref{eq:coro:norm:ex} holds true.
Consider the case $\sum_{j=1}^\m\tau_j=O(1)$. We make use of decomposition \eqref{proof:prop:ex:tau:1} where $U_i$ is replaced by $U_i+n^{-1/2}\delta(Z_i)$. Similarly to the proof of Proposition  \ref{coro:norm:ind:main:tau} it is easily seen that $\sum_{j=1}^\m\tau_j A_{nj}^2 \stackrel{d}{\rightarrow}\sum_{j=1}^\infty\lambda_j^\textsl{e}\,\chi_{1j}^2(\delta_j/\lambda_j^\textsl{e})$. Thereby, due to the proof of Theorem \ref{prop:ex:tau}, the assertion follows.
\end{proof}

\begin{proof}[\textsc{Proof of Proposition \ref{prop:ex:cons}.}] 
Similar to the proof of Proposition \ref{prop:par:cons}.
\end{proof}

\begin{proof}[\textsc{Proof of Proposition \ref{prop:ex:unf:cons}.}]
We make use of inequality \eqref{prop:par:unf:cons:eq} where $\phi(\cdot,\hthet)$ is replaced by $\overline \sol_\k$. 
From the proof of Theorem \ref{prop:ex} we infer $\big\|n^{-1/2}\sum_i(\overline\sol_\k(Z_i)-\sol_0(Z_i))f_\umn^\tau(W_i)\big\|^2=o_p(\varsigma_\m)$ and
\begin{multline*}
 \skalarV{n^{-1}\sum_i(\sol(Z_i)-\overline\sol_\k(Z_i))f_\umn^\tau(W_i),\sum_iU_if_\umn^\tau(W_i)}\\
=\skalarV{n^{-1}\sum_i(\sol(Z_i)-\sol_0(Z_i))f_\umn^\tau(W_i),\sum_iU_if_\umn^\tau(W_i)}+o_p(\varsigma_\m)
\end{multline*}
uniformly over all $\sol\in\cI_n^\sr$. 
Thus, following line by line the proof of Proposition \ref{prop:unf:cons}, the assertion follows.
\end{proof}

\subsection{Proofs of Section 5.}
Recall that $[\op]_\uk=\Ex[e_\uk(W)e_\uk(Z)^t]$. Further,  we denote $[\hop]_\uk=n^{-1}\sum_{i=1}^ne_\uk(W_i)e_\uk(Z_i)^t$ and $[\widehat g]_\uk=n^{-1}\sum_{i=1}^nY_ie_\uk(W_i)$. In the following, we introduce the function $\sol_\k(\cdot):=e_\ukn(\cdot)^t[\op]_\ukn^{-1}[g]_\ukn$ which belongs to $\cL_Z^2$. For all $k\geq 1$ let us denote $\Omega_k:=\{\|[\hop]_\uk^{-1}\|\leq \sqrt n\}$ and $\mho_k:=\{\|R_k\|\|[\op]_\uk^{-1}\|\leq 1/2\}$ where $R_k=[\hop]_\uk-[\op]_\uk$. Note that $\Ex\1_{\Omega_\k^c}=\PP(\Omega_\k^c)=o(1)$ (cf. proof of Proposition 3.1 of \cite{BJ2011}) and, hence $\1_{\Omega_\k}=1+o_p(1)$. For a sequence of weights $\omega=(\omega_j)_{j\geq 1}$ we define the weighted norm $\|\phi\|_\omega=\big(\sum_{j\geq 1}\omega_j(\int_\cZ\phi(z)e_j(z)\nu(z)dz)^2\big)^{1/2}$. 

\begin{proof}[\textsc{Proof of Theorem \ref{prop:np}.}]
For the proof we make use of decomposition \eqref{proof:ex:dec} where the estimator $\phi(\cdot,\hthet)$ is replaced by $\hsol_\k$ given in \eqref{gen:def:est2}.
Consider $III_n$. Observe
\begin{multline}\label{proof:prop:np:dec}
III_n\leq 2\|n^{-1}\sum_i(\sol_\k(Z_i)-\hsol_\k(Z_i))f_\umn^\tau(W_i)\|^2\\
+2\|n^{-1}\sum_i\big(\sol_\k(Z_i)-\sol_0(Z_i)\big)f_\umn^\tau(W_i)\|^2 =2A_{n1}+2A_{n2}.
\end{multline}
Consider $A_{n1}$. Making use of the relation $[\hop]_\ukn[\op]_\ukn^{-1}[g]_\ukn-[\widehat g]_\ukn=n^{-1}\sum_if_\ukn(W_i)(\sol_\k(Z_i)-Y_i)$ we obtain
\begin{multline*}
 A_{n1}\leq 4\big\|\Ex[f_\umn(W)e_\ukn(Z)^t][\op]_\ukn^{-1}\big\|^2\big\|n^{-1}\sum_if_\ukn(W_i)(\sol_\k(Z_i)-Y_i)\big\|^2\\\hfill
+4\big\|\Ex[f_\umn(W)e_\ukn(Z)^t][\op]_\ukn^{-1}R_\k[\hop]_\ukn^{-1}n^{-1}\sum_if_\ukn(W_i)(\sol_\k(Z_i)-Y_i)\big\|^2\\\hfill
+2\|\sol_\k-\hsol_\k\|_\opw^2\sum_{j=1}^\m\tau_j\sum_{l=1}^\k\opw_l^{-1}|n^{-1}\sum_ie_l(Z_i)f_j(W_i)-[\op]_{jl}|^2\\
=4B_{n1}+4B_{n2}+2B_{n3}.
\end{multline*}
From Lemma A.1 of \cite{BJ2011} we deduce $\|n^{-1/2}\sum_ie_\ukn(W_i)\big(\sol_\k(Z_i)-Y_i\big)\|^2=O_p(\k)$ and since $\|\Ex[f_\umn(W)e_\ukn(Z)^t][\op]_\ukn^{-1}\|=O(1)$ we have $nB_{n1}=o(\varsigma_\m)$. Further, consider $B_{n2}$. By employing   $\normV{[\hop]_{\uk}^{-1}}\1_{\mho_k}\leq 2\normV{[\op]_{\uk}^{-1}}$ and
$\normV{[\hop]_{\uk}^{-1}}^2\1_{\Omega_k}\leq n$ for all $k\geq 1$ it follows
\begin{multline*}
 B_{n2}\1_{\Omega_\k}(\1_{\mho_\k}+\1_{\mho_\k^c}) 
=O\Big(4\|[\op]_\ukn^{-1}\|^2\|R_\k\|^2\|n^{-1}\sum_if_\ukn(W_i)\big(\sol_\k(Z_i)-Y_i\big)\|^2\\
+n\|R_\k\|^2\|n^{-1}\sum_if_\ukn(W_i)\big(\sol_\k(Z_i)-Y_i\big)\|^2\1_{\mho_\k^c}\Big).
\end{multline*}
Further, since $n\|R_\k\|^2=O_p(k_n^2)$ (cf. Lemma A.1 of \cite{BJ2011}) and  $n\|R_\k\|^2\|n^{-1/2}\sum_ie_\ukn(W_i)\big(\sol_\k(Z_i)-Y_i\big)\|^2\1_{\mho_\k^c}=o_p(1)$ (cf.  proof of Proposition 3.1 of \cite{BJ2011}) it follows $nB_{n2}\1_{\Omega_\k}=o(\varsigma_\m)$.
This together with estimate \eqref{proof:prop:par:1} 
implies $nA_{n1}=o_p(\varsigma_\m)$. 
Consider $A_{n2}$. We observe
\begin{multline}
\Ex A_{n2} \leq 2\|T(\sol_\k-\sol_0)\|_W^2+2n^{-1}\Ex\|\big(\sol_\k(Z)-\sol_0(Z)\big) f_\umn^\tau(W)\|^2\\\hfill
\leq 2\eta_pd\|\sol_\k-\sol_0\|_\opw^2+\frac{2}{n}\sum_{l\geq 1}l^2\Big(\int_\cZ(\sol_\k-\sol_0)(z)e_l(z)\nu(z)dz\Big)^2\sum_{j=1}^\m\tau_j\sum_{l\geq 1}l^{-2}\Ex|e_l(Z)f_j(W)|^2\\
=O\Big(\frac{\opw_\k}{\sw_\k}\|\sol_\k-\sol_0\|_\sw^2+ \|\sol_\k-\sol_0\|_\sw^2\frac{ k_n^2}{n\sw_\k}\sum_{j=1}^\m\tau_j\Big).
\end{multline}
where we used Lemma A.2 of \cite{Johannes2009}, i.e., $\|\sol_\k-\sol_0\|_w^2=O( w_\k \sw_\k^{-1})$ for a nondecreasing sequence $w$. Condition \eqref{cond:prop:np} together with the estimate $k_n^2\leq \sigma^4\sum_{j=1}^\m\tau_j$ for $n$ sufficiently large implies $nA_{n2}=o_p(\varsigma_\m)$. Consequently, due to \eqref{proof:prop:np:dec} we have shown $n III_n=o_p(\varsigma_\m)$.
The proof of  $n II_n=o_p(\varsigma_\m)$ is based on decomposition \eqref{proof:ex:an:dec:bn} where $\osol_\k$ and $E_\k^\perp\sol_0$ are replaced by $\hsol_\k$
and $\sol_\k-\sol_0$, respectively. Consider $C_{n1}$. We calculate
\begin{equation*}
 C_{n1}\leq \|\hsol_\k-\sol_\k\|_\opw \sum_{j=1}^\m\tau_j\big|\sum_iU_if_j(W_i)\big|\Big(\sum_{l=1}^\k\opw_l^{-1}\big|n^{-1}\sum_ie_l(Z_i)f_j(W_i)-[\op]_{jl}\big|^2\Big)^{1/2}
\end{equation*}
Since $\sqrt n\|\hsol_\k-\sol_\k\|_\opw=o_p(\varsigma_\m^{1/2})$ we obtain, similarly as in the proof of Theorem \ref{prop:ex}, 
 $C_{n1}=o_p(\varsigma_\m)$. 
Consider $C_{n2}$. Again similarly to the  proof of Theorem \ref{prop:ex} we observe
\begin{multline*}
 C_{n2}=\Big|\sum_{j=1}^\m\tau_j\sum_{l=1}^\k[\op]_{jl}\int_\cZ(\hsol_\k-\sol_\k)(z)e_l(z)\nu(z)dz\Big(\sum_iU_if_j(W_i)\Big)\Big|\\\hfill
 \leq \big(n\|\hsol_\k-\sol_\k\|_\opw^2\big)^{1/2}\Big(\sigma^2\sum_{l=1}^\k\opw_l^{-1}\sum_{j=1}^\m[\op]_{jl}^2\Big)^{1/2}+o_p(1)
=o(\varsigma_\m)
\end{multline*}
by exploiting $\sum_{j=1}^\m[\op]_{jl}^2\leq \eta_p\|Te_l\|_W^2\leq d\,\eta_p^2\opw_l$.
Consider $C_{n3}$. Since $\Ex[U^2|W]\leq \sigma^2$ we conclude similarly as in inequality \eqref{proof:ex:ineq:2} using Lemma A.2 of \cite{Johannes2009}
\begin{equation*}
 \Ex C_{n3}\leq \sigma\sum_{j=1}^\m\tau_j \big(\Ex |(\sol_\k(Z)-\sol_0(Z)) f_j(W)|^2\big)^{1/2}
\leq \eta^2\frac{\pi\sigma}{\sqrt 6}\frac{\k}{\sqrt {\sw_\k}}\|\sol_\k-\sol_0\|_\sw\sum_{j=1}^\m\tau_j=o(\varsigma_\m).
\end{equation*}
Consider $C_{n4}$. Again exploring the link condition $\op\in\cTdDw$ and Lemma A.2 of \cite{Johannes2009} we calculate
\begin{multline*}
 \Ex|C_{n4}|^2\leq n\sigma\sum_{j=1}^\m [T(\sol_\k-\sol_0)]_j^2\leq n\sigma\|T(\sol_\k-\sol_0)\|_W^2\\
\leq n\sigma d\|\sol_\k-\sol_0\|_\opw^2\leq 4D d \sr\sigma \frac{n\opw_\k}{\sw_\k}\|\sol_\k-\sol_0\|_\sw^2=o(\varsigma_\m).
\end{multline*}
Consequently, the estimates for $C_{n1}$, $C_{n2}$, $C_{n3}$, and $C_{n4}$ imply $nII_n=o_p(\varsigma_\m)$, which completes the proof.
\end{proof}

\begin{proof}[\textsc{Proof of Theorem \ref{prop:np:tau}.}]
Observe $[\hop]_\ukn[\op]_\ukn^{-1}[g]_\ukn-[\widehat g]_\ukn=n^{-1}\sum_ie_\ukn(W_i)(\sol_\k(Z_i)-Y_i)$  and hence, for all $1\leq j\leq \m$
\begin{multline}\label{prop:np:tau:eq:1}
 n^{-1/2}\sum_if_j(W_i)\big(U_i+\sol_0(Z_i)-\hsol_\k(Z_i)\big)\\\hfill
=n^{-1/2}\sum_i\Big(f_j(W_i)U_i+\Ex\big[f_j(W)e_\ukn(Z)^t\big][\op]_\ukn^{-1}e_\ukn(W_i)\big(\sol_\k(Z_i)-Y_i\big)\Big)\\\hfill
-n^{-1/2}\sum_i \Ex\big[f_j(W)e_\ukn(Z)^t\big][\op]_\ukn^{-1}R_\k[\hop]_\ukn^{-1} e_\ukn(W_i)\big(\sol_\k(Z_i)-Y_i\big)\\
+\Big(n^{-1}\sum_if_j(W_i) e_\ukn(Z_i)^t-\Ex\big[f_j(W)e_\ukn(Z)^t\big]\Big) [\hop]_\ukn^{-1}\Big(n^{-1/2}\sum_ie_\ukn(W_i)\big(\sol_\k(Z_i)-Y_i\big)\Big)\\
+ n^{-1/2}\sum_i \big(\sol_0(Z_i)-\sol_\k(Z_i)\big)f_j(W_i)=A_{nj}+B_{nj}+C_{nj}+D_{nj}.
\end{multline}
Consider $A_{nj}$. 
For each $j\geq 1$, note that $\|\Ex\big[f_j(W)e_\uk(Z)^t\big][\op]_\uk^{-1}\|$ is bounded uniformly in $k$ and further that $\Ex[e_\ukn(W)(\sol_\k(Z)-\sol_0(Z))]=0$.
Now similarly to the proof of Theorem \ref{prop:ex:tau} we conclude
\begin{equation*}
 \sum_{j=1}^\m \tau_jA_{nj}^2 
= \sum_{j=1}^\m\tau_j\Big|n^{-1/2}\sum_iU_i\big(f_j(W_i)-e_j(W_i)\big)\Big|^2+o_p(1)
\stackrel{d}{\rightarrow}\sum_{j=1}^\infty\lambda_j^\textsl{np}\,\chi_{1j}^2.
\end{equation*}
Moreover, as in the  proof of Theorem \ref{prop:np} it can be seen that $\sum_{j=1}^\m \tau_j (B_{nj}^2+C_{nj}^2+D_{nj}^2)=o_p(1)$, which proves the result
\end{proof}

\begin{proof}[\textsc{Proof of Proposition  \ref{coro:norm:np}.}]
Consider the case $\varsigma_\m^{-1}=o(1)$.
Further, under \eqref{loc:alt:np} we observe by following the upper bound for $A_{n1}$ in the proof of Theorem \ref{prop:np} that
\begin{multline*}
 \sum_{j=1}^\m\Big|n^{-1/2}\sum_i(\hsol_\k(Z_i)- \sol(Z_i))f_j^\tau(W_i)\Big|^2\\
=\sum_{j=1}^\m\frac{1}{n}\Big|\sum_i(\hsol_\k(Z_i)-\sol_\k(Z_i))f_j^\tau(W_i)\Big|^2+\sum_{j=1}^\m\Big|\frac{\varsigma_\m}{n}\sum_{i}\delta(Z_i)f_j^\tau(W_i)\Big|^2+o_p(\varsigma_\m)\\
=\varsigma_\m\sum_{j=1}^\m\delta_j^2+o_p(\varsigma_\m).
\end{multline*}
Consequently, the result follows as in the proof of Theorem \ref{prop:np}.
For $\sum_{j=1}^\m\tau_j=O(1)$ we conclude similarly.
\end{proof}

\begin{proof}[\textsc{Proof of Proposition \ref{prop:np:cons}.}] 
Following the lines of the proof of Theorem \ref{prop:np} it can be seen that $\|n^{-1}\sum_i(\hsol_\k(Z_i)-\sol_\k(Z_i))f_\umn^\tau(W_i)\|^2=o_p(1)$ with $\sol\in \cF_\sw^\rho$. On the other hand, $\|n^{-1}\sum_i(\sol_\k(Z_i)-\sol(Z_i))f_\umn^\tau(W_i)\|^2=C\|T(\sol_\k-\sol)\|_W^2+o_p(1)$. Further, since $\|T(\sol_\k-\sol)\|_W^2=\|g-T\sol\|_W^2+o(1)$ the result follows as in the proof of Proposition \ref{prop:par:cons}.
\end{proof}

\begin{proof}[\textsc{Proof of Proposition \ref{prop:np:unf:cons}.}]
We make use of inequality \eqref{prop:par:unf:cons:eq} where $\phi(\cdot,\hthet)$ is replaced by $\hsol_\k$. 
From the proof of Theorem \ref{prop:np} we infer $\big\|n^{-1/2}\sum_i(\hsol_\k(Z_i)-\sol_0(Z_i))f_\umn^\tau(W_i)\big\|^2=o_p(\varsigma_\m)$ and
\begin{multline*}
 \skalarV{n^{-1}\sum_i(\sol(Z_i)-\hsol_\k(Z_i))f_\umn^\tau(W_i),\sum_iU_if_\umn^\tau(W_i)}\\
=\skalarV{n^{-1}\sum_i(\sol(Z_i)-\sol_0(Z_i))f_\umn^\tau(W_i),\sum_iU_if_\umn^\tau(W_i)}+o_p(\varsigma_\m)
\end{multline*}
uniformly over all $\sol\in\cJ_n^\sr$. Consequently, following line by line the proof of Proposition \ref{prop:unf:cons}, the assertion follows.
\end{proof}
\subsection{Technical assertions.}
 Let us introduce $X_{ii'}:=\sqrt 2(\varsigma_\m n)^{-1}\sum_{j=1}^\m U_iU_{i'}f_j^\tau(W_i)f_j^\tau(W_{i'})$ and 
\begin{equation}\label{def:Q}
Q_{ni}:=
\left\{\begin{array}{lcl} 
\sum_{l=1}^{i-1}X_{li}, && \mbox{ for } i=2,\dots, n,\\
0,&&\mbox{ for } i=1 \mbox{ and } i>n.
\end{array}\right.
\end{equation}
Then clearly
\begin{multline*}
(\sqrt 2\varsigma_\m n)^{-1}\sum_{i\neq i'}\sum_{j=1}^\m U_iU_{i'}f_j^\tau(W_i)f_j^\tau(W_{i'})
=\sqrt 2(\varsigma_\m n)^{-1}\sum_{i< i'}\sum_{j=1}^\m U_iU_{i'}f_j^\tau(W_i)f_j^\tau(W_{i'})\\
=\sum_{i<i'}X_{ii'}=\sum_{i=1}^nQ_{ni}.
\end{multline*}
Let $\cB_{ni}:=\cB((Z_1,Y_1,W_1),\dots, (Z_i,Y_i,W_i))$, $1\leq i\leq n$, $n\geq 1$, be the $\sigma$-algebra generated by $(Z_1,Y_1,W_1),\dots, (Z_i,Y_i,W_i)$. Since $U_if_j^\tau(W_i)$, $1\leq i\leq n$, are centered random variables it follows that $\{(\sum_{i'=1}^iQ_{ni'},\cB_{ni}), \,\,i\geq 1\}$ is a Martingale for each $n\geq 1$ and hence $\{(Q_{ni},\cB_{ni}), \,i\geq 1\}$ is a Martingale difference array for each $n\geq 1$. Moreover, it satisfies the conditions of Proposition \ref{prop:adnan} as shown in the following technical result.
\begin{prop}\label{prop:adnan}
 If $\{(Q_{ni},\cB_{ni}),\,i\geq 1\}$ is a Martingale difference array for each $n\geq 1$ satisfying conditions 
\begin{gather}
 \sum_{i= 1}^\infty\Ex|Q_{ni}|^2\leq 1 \quad \text{ for all }n\geq 1\label{normal:lem:1:eq1},\\
 \sum_{i=1}^\infty Q_{ni}^2=1+o_p(1)\label{normal:lem:1:eq2},\\
 \sup_{i\geq 1}|Q_{ni}|=o_p(1) \label{normal:lem:1:eq3}
\end{gather}
 then 
 $\sum_{i=1}^\infty Q_{ni}\stackrel{d}{\rightarrow} N(0,\nu)$.
\end{prop}
\begin{proof}
 See \cite{Awad}.
\end{proof}
Note that this result has been also applied by \cite{Ghorai} to establish asymptotic normality of an orthogonal series type density estimator.

\begin{lem}\label{normal:lem:1}
 Let $Q_{ni}$ be defined as in \eqref{def:Q}. Let  Assumptions 1--4 be satisfied and assume $\big(\sum_{j=1}^\m\tau_j\big)^3=o(n)$. Then conditions \eqref{normal:lem:1:eq1}--\eqref{normal:lem:1:eq3} hold true.
\end{lem}
\begin{proof}
Proof of \eqref{normal:lem:1:eq1}. Observe that $\Ex[X_{1i}X_{1i'}]=0$ for $i\neq i'$ and thus, for $i=2,\dots,n$ we have
\begin{multline*}
\Ex |Q_{ni}|^2=\Ex|X_{1i}+\dots+X_{i-1,i}|^2=(i-1)\Ex|X_{12}|^2=\frac{2(i-1)}{n^2\varsigma_\m ^2}\Ex\big|\sum_{j=1}^\m U_1f_j^\tau(W_1)U_2f_j^\tau(W_2)\big|^2\\
=\frac{2(i-1)}{n^2\varsigma_\m ^2}\sum_{j,j'=1}^\m\big(\Ex[ U^2f_j^\tau(W)f_{j'}^\tau(W)]\big)^2
=\frac{2(i-1)}{n^2} 
\end{multline*}
by the definition of $\varsigma_\m $. Thereby, we conclude
\begin{equation}\label{normal:lem:proof:eq1}
 \sum_{i= 1}^n\Ex|Q_{ni}|^2=\frac{2}{n^2}\sum_{i= 1}^{n-1}i=\frac{n(n-1)}{n^2}=1-\frac{1}{n}
\end{equation}
which proves \eqref{normal:lem:1:eq1}.

Proof of \eqref{normal:lem:1:eq2}. 
Using relation \eqref{normal:lem:proof:eq1} we observe
\begin{equation*}
 \Ex\big|\sum_{i=1}^nQ_{ni}^2-1\big|^2=\sum_{i=1}^n \Ex Q_{ni}^4+2\sum_{i< i'}\Ex Q_{ni}^2Q_{ni'}^2-1+o(1)=:I_n+II_n-1+o(1).
\end{equation*}
Consider $I_n$. 
Observe that
\begin{multline*}
 \Ex |Q_{ni}|^4=\Ex\big|\sum_{i'=1}^{i-1}X_{i'i}\big|^4=\Ex\Big|\frac{\sqrt 2}{n\varsigma_\m }\sum_{j=1}^\m\tau_j U_if_j(W_i)\sum_{i'=1}^{i-1}U_{i'}f_j(W_{i'})\Big|^4\\\hfill
  \leq\frac{4}{n^4\varsigma_\m ^4}\Big(\sum_{j=1}^\m\tau_j\Big)^3\sum_{j=1}^\m \Ex |Uf_j(W)|^4\Big((i-1)\tau_j\Ex|Uf_j(W)|^4+3(i-1)(i-2)\varsigma_{jj}^2\Big)
\end{multline*}
where we used that $\Ex[ Uf_j(W)]=0$. Since  $\sum_{i=1}^n3(i-1)(i-2)=n(n-1)(n-2)$ we conclude
\begin{equation*}
 I_n \leq \frac{4}{n^4\varsigma_\m ^4}\Big(\sum_{j=1}^\m\tau_j\Big)^3\Big(\frac{n(n-1)}{2}\sum_{j=1}^\m \tau_j(\Ex |Uf_j(W)|^4)^2+n(n-1)(n-2)\sum_{j=1}^\m {\textsl s}_{jj}^2\Ex |Uf_j(W)|^4\Big)
\end{equation*}
Therefore, applying $\max_{j\geq 1}\Ex |Uf_j(W)|^4\leq \eta_f\,\eta_p \,\sigma^{4}$ and $\sum_{j=1}^\m\tau_j=o(n^{1/3})$ yields $I_n=o(1)$.  Consider $II_n$. 
 We calculate  for $i<i'$
\begin{multline*}
 Q_{ni}^2Q_{ni'}^2
=\Big(\sum_{k=1}^{i-1}X_{ki}^2\Big)\Big(\sum_{k=1}^{i'-1}X_{ki'}^2\Big)
+\Big(\sum_{k=1}^{i-1}X_{ki}^2\Big)\Big(\sum_{k\neq k'}^{i'-1}X_{ki'}X_{k'i'}\Big)\\\hfill
+\Big(\sum_{k\neq k'}^{i-1}X_{ki}X_{k'i}\Big)\Big(\sum_{k=1}^{i'-1}X_{ki'}^2\Big)
+\Big(\sum_{k\neq k'}^{i-1}X_{ki}X_{k'i}\Big)\Big(\sum_{k\neq k'}^{i'-1}X_{ki'}X_{k'i'}\Big)\\
=:A_{ii'}+B_{ii'}+C_{ii'}+D_{ii'}.
\end{multline*}
Consider $A_{ii'}$. Exploiting relation \eqref{normal:lem:proof:eq1} and using
 $\sum_{i<i'}(i-1)
=\sum_{i'=1}^n(i'-1)(i'-2)/2=n(n-1)(n-2)/6$ and further  $\sum_{i<i'}(i-1)(i'-3)=\sum_{i'=1}^n(i'-3)(i'-2)(i'-1)/2=n(n-1)(n-2)(n-3)/8$  we obtain 
\begin{multline*}
 2\sum_{i<i'}\Ex A_{ii'}=4 \Ex X_{12}^2X_{23}^2\sum_{i<i'}(i-1) +2(\Ex X_{12}^2)^2\sum_{i<i'}(i-1)(i'-3)+o(1)\\
=\frac{8n(n-1)(n-2)}{3n^4\varsigma_\m^4}\Big(\sum_{j,j',l,l'=1}^\m\varsigma_{jj'}\varsigma_{ll'}\Ex U^4f_j^\tau(W)f_{j'}^\tau(W)f_l^\tau(W)f_{l'}^\tau(W)\Big)\\+\frac{n(n-1)(n-2)(n-3)}{n^4}+o(1).
\end{multline*}
Moreover, applying the Cauchy-Schwarz inequality twice gives
\begin{multline*}
 \sum_{j,j',l,l'=1}^\m{\textsl s}_{jj'}{\textsl s}_{ll'}\Ex U^4f_j^\tau(W)f_{j'}^\tau(W)f_l^\tau(W)f_{l'}^\tau(W)
\leq \max_{1\leq j\leq \m}\Ex|Uf_j(W)|^4 \Big(\sum_{j,\,j'=1}^\m\sqrt{\tau_j\tau_{j'}}{\textsl s}_{jj'}\Big)^2\\
\leq \eta_f\,\eta_p\,\sigma^4 \varsigma_\m^2\Big(\sum_{j=1}^\m\tau_j\Big)^2.
\end{multline*}
Thereby, it holds $2\sum_{i<i'}\Ex A_{ii'}=1+o(1)$. 
Now consider $B_{ii'}$. Since $\{\basW_l\}_{l\geq 1}$ forms an orthonormal basis on the support of $W$ we obtain
\begin{multline*}
 \Ex \Big(\sum_{k=1}^{i-1}X_{ki}^2\Big)\Big(\sum_{k\neq k'}^{i'-1}X_{ki'}X_{k'i'}\Big)=2\sum_{k=1}^{i-1}\Ex X_{ki}^2X_{ki'}X_{ii'}\\
\leq \frac{8(i-1)}{n^4\varsigma_\m ^4}\sum_{j,j'=1}^\m\Ex\Big| U_1^3f_j^\tau(W_1)f_{j'}^\tau(W_1) U_2^3f_j^\tau(W_2)f_{j'}^\tau(W_2)\sum_{l,l'=1}^\m\varsigma_{ll'} f_l^\tau(W_1) f_{l'}^\tau(W_2)\Big|\\
\leq \frac{8(i-1)\sigma^2 \eta_p^2}{n^4\varsigma_\m ^3}
\Big(\sum_{j,j'=1}^\m\Ex| U^2f_j^\tau(W)f_{j'}^\tau(W)|^2\Big)\leq \frac{8(i-1)\sigma^6 \eta_f\, \eta_p^3}{n^4\varsigma_\m ^3}\Big(\sum_{j=1}^\m\tau_j\Big)^2.
\end{multline*}
This, together with relation \eqref{normal:lem:proof:eq1}, yields $\sum_{i<i'}\Ex B_{ii'}=o(1)$.
Further, it is easily seen that $\sum_{i<i'}\Ex C_{ii'}=o(1)$. 
Consider $D_{ii'}$. Using twice the law of iterated expectation gives
\begin{multline*}
 \Ex D_{ii'}=\Ex\Big(\sum_{k\neq k'}^{i-1}X_{ki}X_{k'i}\Big)\Big(\sum_{k\neq k'}^{i'-1}X_{ki'}X_{k'i'}\Big)
  =4 \sum_{k<k'}^{i-1}\Ex X_{ki}X_{k'i}X_{ki'}X_{k'i'}\\
  =4 \sum_{k<k'}^{i-1}\Ex\big[ X_{ki}X_{k'i}\Ex[X_{ki'}X_{k'i'}|(Y_k,Z_k,W_k),(Y_{k'},Z_{k'},W_{k'}),(Y_{i},Z_{i},W_{i})]\big]\\
  =\frac{8}{n^2\varsigma_\m^2} \sum_{k<k'}^{i-1}\Ex\Big[ \Ex[X_{ki}X_{k'i}|(Y_k,Z_k,W_k),(Y_{k'},Z_{k'},W_{k'})]\sum_{j,j'=1}^\m{\textsl s}_{jj'}U_kf_j^\tau(W_k)U_{k'}f_{j'}^\tau(W_{k'})\Big]\\
  =\frac{8}{n^4\varsigma_\m ^4}\Ex\Big|\sum_{j,j'=1}^\m{\textsl s}_{jj'}U_1f_j^\tau(W_1)U_{2}f_{j'}^\tau(W_{2})\Big|^2(i-1)(i-2)
  \leq \frac{8\sigma^4\eta_p^2}{n^4\varsigma_\m ^2}(i-1)(i-2).
\end{multline*} 
Since  $\varsigma_\m^{-1}=o(1)$ we obtain
\begin{equation*}
 \sum_{i<i'}\Ex D_{ii'}\leq \frac{8\sigma^4\eta_p^2}{n^4\varsigma_\m ^2}\sum_{i<i'}(i-1)(i-2)
  = \frac{2\sigma^4\eta_p^2\,n(n-1)(n-2)(n-3)}{3\varsigma_\m^{2}n^4}=o(1)
\end{equation*} 
and hence $2\sum_{i< i'}\Ex Q_{ni}^2Q_{ni'}^2=1+o(1)$.

Proof of \eqref{normal:lem:1:eq3}. Note that $\PP\big(\sup_{i\geq 1}|Q_{ni}|>\varepsilon\big)\leq \sum_{i=1}^n\PP\big(Q_{ni}^2>\varepsilon^2\big)$ and, hence the assertion follows from the Markov inequality.
\end{proof}
\bibliography{BiB-NIR-LinFun}
\end{document}